\documentclass[12pt,oneside]{article}

\usepackage{amsmath}
\usepackage{amsthm}
\usepackage{amssymb}
\usepackage{graphicx}

\numberwithin{equation}{section}

\newtheorem{theorem}{Theorem}
\newtheorem{lemma}[theorem]{Lemma}
\newtheorem{proposition}[theorem]{Proposition}

\theoremstyle{definition}

\theoremstyle{remark}

\usepackage[utf8]{inputenc}
\usepackage[T1]{fontenc}

\usepackage{hyperref}
\hypersetup{
  colorlinks=true,
  linkcolor=red,        % For equations, sections, etc.
  citecolor=blue,       % For citations
  urlcolor=blue
}

\begin{document}

\title{Chern-Simons deformations of the gauged O(3) Sigma model on compact surfaces}
\author{René I. García-Lara\thanks{E-mail: \texttt{israel.garcia@im.unam.mx}}}
\date{Instituto de Matemáticas,\\
	Universidad Nacional Autónoma de México\\
	Avenida Universidad s/n, 62210,\\
	Cuernavaca, Mexico \\
\today
}

\maketitle

\begin{abstract}
	Existence of solutions to the field equations of the gauged 
	Chern-Simons-O(3)-Sigma model on a compact Riemann surface 
	is proved by a topological method. 
	Existence of a minimal deformation constant \(\kappa_{*} > 0 \) 
	is proved, such that for any prescribed configuration of vortices and antivortices, at least one solution exists for \(|\kappa| \leq
	\kappa_{*} \). 
	For small values of the Chern-Simons deformation parameter $\kappa$, 
	it is proved that the field equations admit multiple solutions, 
	provided the total number of vortices and antivortices are different. 
	The Maxwell limit is computed for solutions of the 
	field equations. 
	In contrast, if the number of vortices equals the number 
	of antivortices, it is proved that the field equations admit at least one solution for any value of $\kappa$ and the limit $\kappa \to \infty$ is 
	proved. Dependence of the fields on the deformation parameter is
	investigated numerically on the sphere. 
\end{abstract}


% MSC 
% 35A01 Existence problems for PDEs: global existence, local existence, non-existence
% 35A16 Topological and monotonicity methods applied to PDEs
%70S15  Yang-Mills and other gauge theories in mechanics of particles and systems
%35Q51  Soliton equations 
%35Q60  PDEs in connection with optics and electromagnetic theory
%47H11  Degree theory for nonlinear operators
%58Cxx 	Calculus on manifolds; nonlinear operators

\section{Introduction}\label{sec:intro}

The \(O(3)\)-Sigma model is a well known field theory model in \(2 + 1\) dimensions for a matter field coupled to a gauge field with \(U(1)\) symmetry~\cite{schroersBogomolnyiSolitonsGauged1995b}. What makes this model interesting is that it admits topological soliton solutions, determined by a finite set of core points on space of two different types, the vortex and antivortex positions. The moduli space of vortices-antivortices is the finite dimensional space of gauge equivalence classes of solutions to the field equations determined by the core set. A thorough study of the moduli space can be found in reference~\cite{romaoGeometrySpaceBPS2020}.  
Let \((\Sigma, g) \) be a surface with a Riemannian metric, equipped also with a \( U(1) \to P \to \Sigma \) principal bundle. Let \( \rho : U(1) \to SO(3) \) be a representation of \( U(1) \) as rotations with respect to a fixed axis and let \( \ell = P \times_{\rho} \mathbb{S}^2 \) be a complex line bundle with
\( U(1) \) symmetry. On each fibber \( \ell_p \) we consider a Riemannian metric
compatible with the almost complex structure induced by \( \rho \). The representation induces two smooth sections representing its fixed points, we choose arbitrarily one of them and call it \( n \), the north pole section. 
Extend the metric to the Lorentzian manifold \(M = \mathbb{R} \times \Sigma \) as \(\tilde{g} = -dt^2 + g \), if \( \mathcal{A} \) is a time varying connection, locally represented in an open set \(I \times U \) by a connection form \( a \in \Omega^1(I \times U) \), and \( \phi: I \times U \to \ell \) is a local section, the covariant derivative induced by the connection is, 
\begin{equation}
	D \phi = d\phi - a \otimes X_{\phi},
\end{equation}
where \( X \in \mathfrak{X}(\mathbb{S}^2) \) is the canonical generator of the
representation of \( \mathfrak{u}(1) \) induced by \( \rho \), and \( X_{\phi}
\in T_{\phi}\ell \) is the pullback. Recall the curvature form is the globally defined 2-form \(F = da \), which decomposes as \(F = e\wedge dt + B \), where \(e \) is the electric 1-form and  \(B \) is the magnetic 2-form, both defined globally. 
The \(O(3)\)-Sigma model Lagrangian is,
\begin{equation*}
	\mathcal{L}_{O(3)} = \frac{1}{2} \langle D_{\mu}\phi, D^{\mu}\phi \rangle - \frac{1}{4 q} F_{\mu\nu}F^{\mu\nu} 
	- \frac{q}{2} {\left( 
		 \tau - \langle n, \phi \rangle 
		\right)}^2.
\end{equation*}
Brackets refer to the inner product on \(\ell \) induced by the round metric on \(\mathbb{S}^2\), whereas the parameters satisfy \(q > 0 \) and \(-1 < \tau < 1 \).

Ghosh and Ghosh added a Chern-Simons term to the \(O(3)\)-Sigma model to produce soliton solutions that spin in the internal space~\cite{ghoshTopologicalNontopologicalSolitons1996}. Later, Kimm, Lee and Lee introduced a neutral scalar field \(\mathcal{N}\) in order to make the following Chern-Simons deformation self-dual~\cite{kimmAnyonicBogomolnyiSolitons1996}, 
\begin{equation}
	\mathcal{L} = \frac{1}{2} \langle D_{\mu}\phi, D^{\mu}\phi \rangle - \frac{1}{4 q} F_{\mu\nu}F^{\mu\nu} 
	+ \frac{1}{2q} \partial_{\mu}\mathcal{N} \partial^{\mu}\mathcal{N} + \frac{\kappa}{2} * (a \wedge F) 
	- V(\phi, \mathcal{N}),
\end{equation}
where the potential function \(V \) is 
\begin{equation}
	V(\phi, \mathcal{N}) = \frac{q}{2} {\left( 
		\kappa \mathcal{N} + \tau - \langle n, \phi \rangle 
		\right)}^2 + \frac{1}{2} |\mathcal{N} X_{\phi}|^2.
\end{equation}
 If \(\kappa = 0\), \(\mathcal{L}\) reduces to \(\mathcal{L}_{O(3)}\), whereas for \(\kappa \neq 0\), it introduces a Chern-Simons term. This term is not gauge invariant; however, any two gauge related terms differ by a divergence. 
Variations with respect to \(a_{t} \) yield Gauss law, 
\begin{equation}
	\frac{1}{q} \delta e = \langle D_{t} \phi, X_{\phi} \rangle + \kappa * B, 
\end{equation}
where \(\delta = -*d* \) is the co-differential on \(\Sigma \) with respect to \(g \). The standard Bogomol'nyi argument and the Gauss law yield stationary solutions satisfying the self dual equations (see ref.~\cite{chaeCondensateMultivortexSolutions2001}),
\begin{align}
	\partial_{t}{\mathcal{N}}     & = 0, \label{eq:nf}                \\
	e           & = d\mathcal{N}, \label{eq:e}                \\
	D_t\phi     & = \mathcal{N} X_{\phi}, \label{eq:X}        \\
	\bar{\partial}_{\mathcal{A}}\phi & = 0,\label{eq:hf}        \\
	*B & = -q (\kappa \mathcal{N} + \tau - \langle n, \phi \rangle),\label{eq:B}
\end{align}
where \( \bar{\partial}_{\mathcal{A}} \) is the anti-holomorphic part of \(D \) in 
\( T\ell \otimes \mathbb{C} \). Let \( \phi_3 = \langle n, \phi \rangle \) be the north pole projection of a section, the sets of vortices and antivortices are the preimages \( \phi_{3}^{-1}(1) \) and \(\phi_{3}^{-1}(-1) \) respectively. By equation~\eqref{eq:hf}, these sets are finite ~\cite[~p.~8]{sibnerAbelianGaugeTheory2010}. 
If we consider an holomorphic coordinate system \( \varphi_{1}: \mathbb{C} \to U \subset \Sigma \) about \( p \in \phi_{3}^{-1}(\pm 1) \) and a trivialisation \( \varphi_2: \ell_U\to \mathbb{C} \times \mathbb{S}^2 \), such that 
\begin{equation}
	(\mathrm{id}, \phi') = \varphi_2 \circ \phi \circ \varphi_1: \mathbb{C} \to
	\widehat{\mathbb{C}},
\end{equation}
then the degree of the map \( \phi' \) at \( \varphi_1^{-1}(p) \) is independent
of the holomorphic coordinates chosen. We call this the degree of \( \phi \) at \(p \). 

To study the Bogomol'nyi equations, notice that in the gauge \( a_0 = -\mathcal{N} \), the fields become stationary,
\( \partial_{t}\phi = \partial_{t} a = 0 \). Let \( u:\Sigma \to \overline{\mathbb{R}} \) be the function 
\begin{equation}
	u = \log \left(\frac{1 - \phi_3}{1 + \phi_3} \right). \label{eq:fn-u}
\end{equation}
The Bogomol'nyi equations imply \((u, \mathcal{N}) \) is a solution to the distributional PDEs, 
\begin{equation}
	\begin{aligned}
		\nabla^2 u & = 2q \left(
		\kappa \mathcal{N} + \tau + \frac{e^u - 1}{e^u + 1}
		\right) + 4\pi \sum_{p\in\phi_{3}^{-1}(\pm 1)} \deg(p) \delta_p,\\
		\nabla^2 \mathcal{N} & = \kappa q^2 \left(
		\kappa \mathcal{N} + \tau + \frac{e^u - 1}{e^u + 1}\right) + q \frac{4 e^u}{{(e^u + 1)}^2} \mathcal{N}, 
	\end{aligned} \label{eq:field-theory-system}
\end{equation}
where \(\nabla^2 \) is the Laplace-Beltrami operator and the degree \(\deg(p) \)  is positive for vortices and negative for antivortices. We will denote as 
\[ k_{\pm} = \sum_{p \in \phi_{3}^{-1}(\pm 1)} |\deg(p)|  \]
the total number of vortices or antivortices counted with multiplicity.

System~\eqref{eq:field-theory-system} is well known in the literature when the space is \(\mathbb{R}^2\) or the flat torus, while the more general case of a compact surface is barely studied. 
On the euclidean plane, Han and Nam prove~\cite{hanTopologicalMultivortexSolutions2005} that the field equations admit a solution up to upper and lower bounds for \(\kappa\). If there are only vortices or antivortices, there are solutions for any value of \(\kappa\)~\cite{hanExistenceAsymptoticsTopological2011c}. On the contrary, in a flat torus if \(k_{+} \neq k_{-}\), \(\kappa\) is bounded and there are multiple solutions for the governing elliptic problem~\cite{chaeCondensateMultivortexSolutions2001,chiacchio2007multiple}. 

The common approach to study the elliptic equations is by variational methods, keeping \(\kappa\) fixed and tuning \(q\) to obtain pairs of sub and super solutions~\cite{hanMultiplicitySelfDualCondensate2014,ricciardiMultiplicityNonlinearFourthorder2004,chiacchio2007multiple}. The asymptotics of the solutions when \(\kappa \to 0\) (the maxwell limit) and \(q \to \infty\) (the Chern-Simons limit) are also of interest~\cite{ricciardiAsymptoticsNonlinearElliptic2003,huangTopologicalSolutionsVortices2022}. 
Here our approach is different, rather than varying \(q\), we aim to study Chern-Simons deformations inspired in the continuation method proposed in reference~\cite{flood2018chern} for the Abelian-Higgs model. We start with a known solution to the \(O(3)\)-Sigma model and perform a small initial deformation in \(\kappa\); then, with help of  Leray-Schauder theory~\cite[Chp.~3]{chipotHandbookDifferentialEquations2004}, we extend to a continuous deformation of the parameter. In comparison, in reference~\cite{hanCondensateSolutionsSelfdual2019} the authors used Leray-Schauder degree theory to prove existence of solutions to the elliptic equations for \(k_{+} = 0\), \(\tau= 1\) on a torus. 
The main results of the work are summarized in the following theorems.

\begin{theorem}\label{thm:existence}
	If \(\tau \in (-1, 1)\), \(q > 0\) are constants such that,
	\begin{equation}
		- 1 - \tau <
		\frac{2\pi (k_+ - k_-)}{qVol(\Sigma)}
		< 1 - \tau, \label{eq:bradlow-cond-thm}
	\end{equation}
	the following statements hold.
	\begin{enumerate}
		\item For any \(k \geq 0\), there is a family \((\phi_{\kappa}, \mathcal{A}_{\kappa}, \mathcal{N}_{\kappa})\) of solutions to the Bogomol'nyi equations, such that for any \(\epsilon > 0\) there is \(\delta > 0\) such that \(|\kappa| < \delta\) implies 
		\begin{equation} \|(\phi_{3})_{\kappa} - (\phi_{3})_{0}\|_{C^k(\Sigma)} + \|\mathcal{N}_{\kappa}\|_{C^k(\Sigma)} < \epsilon. \label{eq:top-limit-k-0} 
		\end{equation}
		\item There is a constant \(\kappa_{*} > 0\) independent of core positions, such that whenever \(|\kappa| < \kappa_{*}\), there is a gauge equivalence class of solutions to the Bogomol'nyi equations~\eqref{eq:nf}--\eqref{eq:B}. If \(k_{+} \neq k_{-}\), we can take \(\kappa_{*}\) such that there are two gauge equivalence classes of solutions.
		\item If \(k_{+} \neq k_{-}\), given a neighbourhood \(U\) of the core set, for any \(\epsilon > 0\), there is a \(\delta > 0\) such that if \(|\kappa| < \delta\), there is a solution \((\phi_{\kappa}, \mathcal{A}_{\kappa}, \mathcal{N}_{\kappa})\) to the Bogomol'nyi equations, such that 
		\begin{equation}
			\|(\phi_{3})_{\kappa} \mp 1 \|_{C^1(\Sigma\setminus \bar{U})} + \left\|\kappa\mathcal{N}_{\kappa} + \frac{2\pi (k_{+} - k_{-})}{q Vol(\Sigma)} \mp 1 + \tau\right\|_{C^{1}(\Sigma)} < \epsilon, \label{eq:non-top-limit-unbalance}
		\end{equation}
		where the sign is chosen depending whether \(k_{+} - k_{-}\) is positive or negative.
	\end{enumerate}
\end{theorem}

Condition~\eqref{eq:bradlow-cond-thm} is a Bradlow type bound, similar to the constraint naturally appearing in the study of vortices of the Abelian-Higgs model~\cite{bradlowVorticesHolomorphicLine1990}. 
The moduli space of the \(O(3)\)-Sigma model is the space \(M_{0}\) of solutions to the Bogomol'nyi equations at \(\kappa = 0\) modulo gauge equivalence. It is a stratified space, whose connected components \(M_{0}^{(k_{+},k_{-})}\) are  determined by the total number of vortices and antivortices; furthermore, it is a smooth, incomplete Kähler manifold~\cite{romaoGeometrySpaceBPS2020} whose low energy dynamics is governed by geodesic motion (the reader can consult~\cite{mantonTopologicalSolitons2004} for a gentle introduction to the subject). Parts 1 and 2 of Theorem~\ref{thm:existence} claim that for small \(\kappa\), the moduli space of solutions to equations~\eqref{eq:nf}--\eqref{eq:B}, \(M_{\kappa}^{(k_{+},k_{-})}\),  
is a small deformation of \(M_0^{(k_{+},k_{-})}\). Persistence of moduli space for small \(\kappa\) is an assumption implicitly done in the Abelian-Higgs model in order to study low energy dynamics with a Chern-Simons term~\cite{collieDynamicsChernSimonsVortices2008a,alqahtaniRicciMagneticGeodesic2015}. Flood and Speight made this assumption formal for the Abelian-Higgs model~\cite{flood2018chern}. 
Theorem~\ref{thm:existence} extends persistence of moduli space to the \(O(3)\)-Sigma model and opens the possibility to study low-energy dynamics of solutions to the Bogomol'nyi equations with geometric technics in \(M_{\kappa}^{(k_{+},k_{-})}\). 
The second theorem shows that as in the euclidean plane, Chern-Simons deformations of the \(O(3)\)-Sigma model can extend indefinitely in any compact surface. Unlike the euclidean case, this can be done only if the number of vortices and antivortices coincide.

\begin{theorem}\label{thm:limit-infty}
	If \(k_{+} = {k_-}\), for any \(\kappa \in \mathbb{R}\) and \(q > 0\) there is a solution to the Bogomol'nyi equations~\eqref{eq:nf}--\eqref{eq:B}. Moreover, if 
	\((\kappa_{i}, \phi_{i}, A_{i}, \mathcal{N}_{i})\) is a sequence of solutions such that \(|\kappa_{i}| \to \infty \), up to subsequences, one of the following statements hold.
		\begin{enumerate}
		\item For any \(p \geq 2\)  and any neighbourhood \(U\) of the core set, 
		 \[ \lim_{i \to 0} \|(\phi_{3})_{i} \mp 1 \|_{C^1(\Sigma\setminus \bar{U})} + \left\|\kappa_{i} \mathcal{N}_{i} \mp 1 + \tau \right\|_{L^p} = 0. \]
		\item For the function \(v\) defined at~\eqref{eq:vpm} below, there is a constant \(c\) such that 
		\begin{equation} 
			\int_{\Sigma} \frac{e^{c+v}}{(e^{c+v} + 1)^2}\left(
				\frac{e^{c+v} - 1}{e^{c + v} + 1} + \tau 
			\right) = 0 \label{eq:c-cond}
		\end{equation}
		and for any \(k \geq 0\), 
		\[\lim_{i \to \infty }\left\|(\phi_{3})_{i} - \frac{1 - e^{c + v}}{1 + e^{c+v}}\right\|_{C^k(\Sigma)} + \left\|\kappa_{i}\mathcal{N}_{i} + 
		\frac{e^{c+v} - 1}{e^{c + v} + 1} + \tau 
		\right\|_{C^k(\Sigma)} = 0. \]
	\end{enumerate}
\end{theorem}

Previous work using an iterative scheme on the torus proved existence of solutions for \(\kappa\) sufficiently small and \(q\) large enough to guarantee convergence to a solution~\cite{chaeCondensateMultivortexSolutions2001}. The limit as \(|\kappa| \to \infty \) is similar to the Chern-Simons limit~\cite[Thm.~4.1-ii]{chaeCondensateMultivortexSolutions2001}, obtained keeping \(\kappa\) fixed, and to limits obtained as \(\kappa \to 0\), \(\kappa^2q\to \infty\) in reference~\cite{huangBreakdownPointwiseConvergence2026}. It is interesting to note that in both, the Chern-Simons limit and the \(\kappa \to \infty\) limit, there are three possibilities for the solutions of the Bogomol'nyi equations. 

The next sections aim to prove the theorems, the content is as follows. In section~\ref{sec:elliptic-prob}, we transform the distributional equations~\eqref{eq:field-theory-system} into a regular system of elliptic PDEs, then we construct solutions to the regular governing system for small \(\kappa\) and prove that \(\kappa\) can be extended at least to a value \(\kappa_{*}\) independent of vortex-antivortex positions. In Section~\ref{sec:cont-extension}, we show that the small extensions of the previous section can be extended continuously to an unbounded set in Sobolev space using Leray-Schauder's degree theory. We compute the possible limits of diverging sequences of solutions to the regular elliptic problem and prove Theorem~\ref{thm:existence}. In Section~\ref{sec:asympt-large}, we study the limit behaviour of solutions to the regular elliptic problem for large \(\kappa\). Using techniques developed for Maxwell-Chern-Simons theory on the torus~\cite{ricciardiAsymptoticsNonlinearElliptic2003,chiacchio2007multiple}, we bound the growth of \(\kappa \mathcal{N}\) to second order and use this bound to obtain the limit behaviour of solutions to the regular elliptic problem and prove Theorem~\ref{thm:limit-infty}. In Section~\ref{sec:symm-deform-sphere}, we study Chern-Simons deformations on the sphere numerically. Using the method of  pseudo arclength continuation introduced in~\cite{flood2018chern}, we produce two families of solutions to the Bogomol'nyi equations showcasing the limit behaviour described in the theorems.

\section{The elliptic problem}\label{sec:elliptic-prob}

From now onwards, let us denote by \(|\Sigma|\) the Riemannian volume of the surface. 
To study the governing elliptic problem~\eqref{eq:field-theory-system}, 
it is customary to replace this system by a regular system of elliptic PDEs  
introducing a Green function for the Laplacian, \(G(x, y) \), such
that,
\begin{equation}
	\nabla^2_{x} G(x,y) = - \frac{1}{|\Sigma|} + \delta_{y}.
\end{equation}
A proof of the existence of \(G \) can be found
in~\cite{aubin2013some}. For any holomorphic chart \(
\varphi: U\subset \Sigma \to \mathbb{C} \) and any bounded domain \(D
\subset \mathbb{C} \), there exists a smooth function \( \tilde G:
\varphi^{-1}(D) \to \mathbb{R} \) such that, for any \(x, y \in
\varphi^{-1}(D) \), 
\begin{equation}
	G(x, y) = \frac{1}{2\pi}\log \, |\varphi(x) - \varphi(y)| +
	\tilde G(x, y).\label{eq:green-approx} 
\end{equation}
Hence, the function \(\exp(4\pi G) \) is smooth and has a zero
whenever \(x = y \). Define 
\begin{equation}\label{eq:vpm}
	v_{\pm}(x) = 4\pi \sum_{y \in \phi_{3}^{-1}(\pm 1)} |\deg(y)|\,G(x, y), \quad v = v_{+} - v_{-},
\end{equation}
if \((u, \mathcal{N})\) is a solution to~\eqref{eq:field-theory-system}, let 
\(h = u - v \), then \((h, \mathcal{N} ) \) is a solution of the  system,
	\begin{align}
		\nabla^2 h & = 2q \left(
		\kappa \mathcal{N} + f({h + v})
		\right) + A,  \label{eq:h-n-cal}	\\ %\label{eq:elliptic-sys} \\
		\nabla^2 \mathcal{N} & = \kappa q^2 \left(
		\kappa \mathcal{N} + f({h + v})\right) + 2 q f'({h + v}) \mathcal{N}, \label{eq:n-cal}
	\end{align} 
where we have defined \(A = 4\pi (k_{+} - k_{-}) {|\Sigma|}^{-1} \) and 
\begin{equation}
	f(x) = 1 + \tau - \frac{2}{e^{x} + 1}.\label{eq:f-cp1-model}
\end{equation}
Given \(k \in \mathbb{N}_{0}\) and \(p > 0\), we denote by \(W^{k,p}\) the Sobolev space of \(L^p\) functions on \(\Sigma\) with \(k\) weak derivatives in \(L^p\). Notice that \(H^k = W^{k,2}\) by definition. 
For \(\kappa = 0\), the unique solution to~\eqref{eq:n-cal} is \(\mathcal{N} \equiv 0\), since the operator \(-\nabla^2 + 2qf'(h+v): H^2 \to L^2 \) is positive;  furthermore, if we rescalate the metric of \(\Sigma\) from \(g\) to \(g_{q} = q g\),~\eqref{eq:h-n-cal} reduces to the equation,
\begin{equation}
\nabla^2_{q}h = 2f(h+v) + \frac{4\pi (k_{+} - k_{-})}{{|\Sigma|}_{q}}, \label{eq:gov-elliptic-o3}
\end{equation}
where \(\nabla^2_{q} = q^{-1}\nabla^2\) is the Laplace-Beltrami operator for \(g_{q}\) and \({|\Sigma|}_{q} = q |\Sigma| \) is the volume of \(\Sigma\) in this metric. Equation~\eqref{eq:gov-elliptic-o3} is the governing elliptic problem for the O(3)-Sigma mode~\cite{sibnerAbelianGaugeTheory2010}. It admits a unique solution \(h_{0}\) if and only if the condition 
\begin{equation}
-1 - \tau < \frac{2\pi(k_{+} - k_{-})}{{|\Sigma|}_{q}} < 1 - \tau \label{eq:bradlow-o3}
\end{equation}
holds~\cite{Xu2025}. Condition~\eqref{eq:bradlow-o3} is exactly condition~\eqref{eq:bradlow-cond-thm} rewritten in the metric \(g_{q}\). Therefore, for \(\kappa = 0\) there is exactly one solution \((h_{0}, 0)\) to system~\eqref{eq:h-n-cal}--\eqref{eq:n-cal}. In the next section we exploit this observation to find solutions of the system for small \(\kappa\).

\subsection{Small deformations of the O(3) Sigma model}

In this section we aim to prove the first part of Theorem~\ref{thm:existence}. 
Let \(T: \mathbb{R} \times H^{k+2} \times H^{k+2} \to H^k \times H^k
\) be the operator with components 
\begin{align*}
	T_1 & = -\nabla^2h + 2q(\kappa\mathcal{N} + f({h+v})) + A,                        \\
	T_2 & = -\nabla^2\mathcal{N} + \kappa q^2 (\kappa\mathcal{N} + f({h+v})) + 2q f'({h+v})\mathcal{N}.
\end{align*}
\((h, \mathcal{N})\) is a solution 
to~\eqref{eq:h-n-cal}--\eqref{eq:n-cal} in \(H^{k+2}\) if and only if \(T(\kappa, h, \mathcal{N}) = \mathbf{0}\). Assume condition~\eqref{eq:bradlow-cond-thm} holds and let \((h_{0}, 0)\) be the unique solution of the system at \(\kappa = 0\) described in the previous section. The derivative of the restriction \(T|_{k=0}\) at \((h_{0}, 0)\) is
\[
d_{(h_{0},0)}T|_{k=0} = L \oplus L,
\]
where \(L = -\nabla^2 + 2q f'({h_{0}+v}): H^{k+2} \to H^k \) is a
Hilbert space isomorphism because the function \(f'({h_{0}+v})\) is
positive except for the finite set of vortex and antivortex points. 
For the functions \(v_{\pm}\) defined at~\eqref{eq:vpm}, 
\(e^{v_{\pm}}\) is smooth
and vanishes only at vortex or antivortex positions.
Substituting into~\eqref{eq:f-cp1-model} we obtain
\begin{align}
	f(h + v) &= 1 + \tau - \frac{2 e^{v_{-}}}{e^{h} e^{v_{+}} + e^{v_{-}}},
	&
	f'(h + v) &= \frac{2 e^{h} e^{v_{+}} e^{v_{-}}}{\bigl(e^{h} e^{v_{+}} +
		e^{v_{-}}\bigr)^{2}}. \label{eq:f-fp-hv}
\end{align}
For any \(k \ge 2\), the Sobolev space \(H^{k}\) is an algebra, there is a constant \(C(k)\) such that for \(u, v \in H^{k}\), \(\|uv\|_{H^k} \leq C(k) \|u\|_{H^k} \|v\|_{H^k}\)~\cite[Chp.~4]{adamsSobolevSpaces2003}; hence, the exponential map \(h \mapsto e^h\) is well defined and smooth as a map \(H^k \to H^k\)~\cite[pg.~157]{flood2018chern}.  Equation~\eqref{eq:f-fp-hv} shows that \(h \mapsto f(h +v)\) and \(h \mapsto f'(h + v)\) are also smooth as maps \(H^k \to H^k\). The continuous embedding \(H^{k+2} \hookrightarrow H^{k}\) shows that \(T\) is a smooth operator. 
\begin{proposition}\label{prop:small-extension}
	Let \(A, q \in \mathbb{R}\), \(q > 0 \) be constants such that
	\begin{equation}
		f_{\min} < -\frac{A}{2q} < f_{\max}.\label{eq:bradlow-cond}
	\end{equation}
	For any \(k \geq 0\), there exists a neighbourhood \(U \subset H^{k+2} \times H^{k+2} \) of
	\((h_{0}, 0) \) and an open interval \(I\) about \(0\)
	such that, for every \(\kappa \in I\),
	system~\eqref{eq:h-n-cal}--\eqref{eq:n-cal} admits a unique solution \((h_{\kappa}, \mathcal{N}_{\kappa})\) in
	\(U\). Moreover, \(\kappa \mapsto (h_{\kappa}, \mathcal{N}_{\kappa}) \) is of
	class \(C^{\infty}\) as a map \(I \to H^{k+2} \times H^{k + 2}
	\). 
\end{proposition}
\begin{proof}
	\(T \) is a smooth map such that \(d_{(h_{0},0)}T|_{\kappa = 0} \) is a Hilbert space isomorphism and 
	\(T(0, h_{0}, 0) = (0, 0) \). By the implicit function
	theorem, there are an open neighbourhood \(U\) as described in the statement of the proposition and exactly one smooth map \(I \to U\), \(\kappa \mapsto (h_{\kappa}, \mathcal{N}_{\kappa})\), such that \(T(\kappa, h_{\kappa}, \mathcal{N}_{\kappa}) = 0\) for \(\kappa \in I\); equivalently, \((h_{\kappa}, \mathcal{N}_{\kappa})\) is the only solution to~\eqref{eq:h-n-cal}--\eqref{eq:n-cal} in \(U\) with parameter \(\kappa\). 
\end{proof}
Proposition~\ref{prop:small-extension} shows that initial solutions \((h_{0}, 0) \) of system~\eqref{eq:h-n-cal}--\eqref{eq:n-cal} at \(\kappa = 0\) can be extended smoothly to pairs \((h_{\kappa}, \mathcal{N}_{\kappa})\) of solutions for small \(\kappa\). We aim to show the existence of a minimal constant \(\kappa_{*} > 0\) such that for any vortex-antivortex configuration, the one-parameter family of solutions obtained persists at least for \(|\kappa| < \kappa_{*}\). Let \(T_{\kappa}\) be the restriction \(T(\kappa, \cdot, \cdot)\) of the operator \(T\) as a map \(H^2\times H^2 \to L^2 \times L^2\). As the proof of the proposition relies on the implicit function theorem, we can extend the solution of the system as long as \(d_{(h_{\kappa}, \mathcal{N}_{\kappa})}T_{\kappa}\) remains an isomorphism. A direct calculation shows 
\[
d_{(h_{\kappa}, \mathcal{N}_{\kappa})}T_{\kappa} = L \oplus L + P = (L \oplus L) (I + {(L \oplus L)}^{-1}P),
\]
where 
\begin{align}
	L &= -\nabla^2 + 2qf'(h+v), \label{eq:op-l1} \\
	P &= \begin{pmatrix}
		0 & 2 q\kappa \\
		\kappa q^2 f'(h + v) + 2q f''(h+v) \mathcal{N} & \kappa^2q^2
	\end{pmatrix}. \nonumber
\end{align}
Hence, if the perturbation \((L\oplus L)^{-1}P\) is small, \(d_{(h_{\kappa}, \mathcal{N}_{\kappa})}T_{\kappa}\) is invertible with inverse
\begin{equation}
	{(d_{(h_{\kappa}, \mathcal{N}_{\kappa})}T_{\kappa})}^{-1} = {(I + {(L \oplus L)}^{-1}P)}^{-1}{(L \oplus L)}^{-1}, \label{eq:dtk-inv}
\end{equation}
and the solution of~\eqref{eq:h-n-cal}--\eqref{eq:n-cal} can be continued. For \(\|{(L \oplus L)}^{-1}P)\| < 1\) in the \(\sup\) norm,	
\begin{equation}
	{(I + {(L \oplus L)}^{-1}P)}^{-1} = \sum_{n=0}^{\infty} ({(L \oplus L)}^{-1}P)^n; \label{eq:neumann-series}
\end{equation}
furthermore,  
\begin{gather}
	{\|(L\oplus L)^{-1}P\|} \leq 2{\|L^{-1}\|}{\| P\|}, \label{eq:l-l-inv-p-norm} \\
	{\|P\|} \leq 2q|\kappa| + |\kappa| q^2 ({\|f'(x)\|}_{L^\infty} + |\kappa|) + 2q {\|f''(x)\|}_{L^\infty} {\|\mathcal{N}\|}_{L^2}. \label{eq:l-inv-p-norm}
\end{gather}
In the following lemmas, we bound the norm of \(L^{-1}\) and find analytic conditions for the invertibility of \(d_{(h_{\kappa}, \mathcal{N}_{\kappa})}T_{\kappa}\).
Given a pair \((\mathbf{p}, \mathbf{q}) \in \Sigma^{k_+} \times \Sigma^{k_-}\), we will call it  valid if \(\mathbf{p}\) and \(\mathbf{q}\) have pairwise distinct coordinates. In other
words, if vortices and antivortices do not overlap. 

\begin{lemma}\label{lem:pot-bound}
	For any $\epsilon > 0$ there is a constant $C(\epsilon) > 0$,
	such that for any valid configuration \((\mathbf{p}, \mathbf{q}) \in \Sigma^{k_+} \times \Sigma^{k_-} \) of vortices and antivortices 
	and any $h \in H^2$ with $\|{h}\|_{H^2} < \epsilon$,
	\begin{equation*}
		{\langle f'(h + v), 1\rangle}_{L^2} \geq C.
	\end{equation*}	
\end{lemma}

\begin{proof}
	\(f'\) is 
	non negative, hence $\langle{f'(h + v), 1}\rangle_{L^2} \geq 0$. Assume towards a
	contradiction the
	existence of sequences $\{h_n\}$, $\{v_n\}$ with vortices and
	antivortices at positions $(\mathbf{p}_n$, $\mathbf{q}_n)$ such
	that  ${\|h_n\|}_{H^2} < \epsilon$ and  
	\begin{equation*}
		{\langle{f'(h_n + v_{n}), 1}\rangle}_{L^2} \to 0.
	\end{equation*}
	\(\Sigma\) is compact, hence, we can assume $\mathbf{p}_n \to
	\mathbf{p}_* \in \Sigma^{k_+}$ and $\mathbf{q}_n \to \mathbf{q}_*
	\in \Sigma^{k_-}$ together with pointwise convergence $v_n \to
	v_{*}$ except possibly at a finite set of points where any point of the tuple $\mathbf{p}_{*}$ coincide with a point of
	$\mathbf{q}_{*}$. By Sobolev's embedding, 
	\(
	{\|h_{n}\|}_{C^0(\Sigma)} \leq C_{\mathrm{Sob}}\,{\|h_{n}\|}_{H^2} 
	\) 
	for some constant \(C_{\mathrm{Sob}} > 0 \). Hence,
	\begin{equation*}
		0 \leq \langle{f'({|v_{\mathbf{p}_{n}}|} + {|v_{\mathbf{q}_{n}}|} +
			C_{\mathrm{Sob}}\epsilon), 1}\rangle_{L^2} \leq \langle{f'(h_n + v_{n}), 1}\rangle_{L^2} \to 0.
	\end{equation*}
	On the other hand, $f'({|v_{\mathbf{p}_{n}}|} + {|v_{\mathbf{q}_{n}}|} +
	C_{\mathrm{Sob}}\epsilon)$ is a sequence of bounded functions
	converging pointwise to the continuous function $f'({|v_{\mathbf{
				p}_{*}}|} + {|v_{\mathbf{q}_{*}}|} +
	C_{\mathrm{Sob}}\epsilon)$. By the dominated convergence theorem,
	\begin{equation*}
		{\langle{f'({|v_{\mathbf{p}_{*}}|} + {|v_{\mathbf{q}_{*}}|} +
				C_{\mathrm{Sob}}\epsilon), 1}\rangle}_{L^2} = 0,
	\end{equation*}
	a contradiction.
\end{proof}

For the next lemma, recall that in any compact Riemannian manifold \(M\), there is a constant \(C_{P}\), Poincare's constant, such that if \(u \in H^1\) is a function with mean \(\int_{\Sigma} u = 0\), then,
\[ \|u\|_{L^2} \leq C_{P} \|\nabla u\|_{L^2}. \]

\begin{lemma}\label{lem:u-l2-bound}
	Let \(M\) be a compact Riemannian manifold. If \(V \in L^2(M)\) is a function such that \(V \geq 0\) and \(\int_{M} V > 0\), for any \(u \in H^1\),
	\[
	{\|u\|}_{L^2}^2 \leq \alpha ({\|\nabla^2 u\|}_{L^2} + \langle Vu, u\rangle_{L^2}),
	\]
	where 
	\begin{equation}
	\alpha = \max \left( 
	C_{P}^2 \left(1 + 4|M|\frac{{\|V\|}_{L^2}^2}{(\int_{M}V)^2}\right),
	\frac{4|M|}{\int_{M}V}
	\right) \label{eq:alpha}	
	\end{equation}
	and \(C_{P}\) is Poincare's constant.
\end{lemma}
\begin{proof}
	Let \(\bar{u} = {|M|}^{-1}\int_{M}u\), \(\tilde{u} = u - \bar{u}\), then \(\tilde{u}\) and \(\bar{u}\) are \(L^2\)-orthogonal and \(\nabla u = \nabla \tilde{u}\). Hence,
	\[{\|u\|}^2_{L^2} = {\|\tilde{u}\|}_{L^2}^2 + \bar{u}^2 |M| \leq C_{P}^2 {\|\nabla u\|}_{L^2}^2 + \bar{u}^2 |M|. \]
	On the other hand,
	\begin{align*}
	\langle Vu, u \rangle_{L^2} &=  \int_{M}V{\tilde{u}}^2 + 2 \bar{u} \int_{M} V\tilde{u} + \bar{u}^2 \int_{M}V \\
	&\geq 2 \bar{u} \int_{M} V\tilde{u} + \bar{u}^2 \int_{M}V\\
	&\geq \bar{u}^2\int_{M}V - 2 {\|V\|_{L^2}}{\|\tilde{u}\|}_{L^2} |\bar{u}|.
	\end{align*}
	This quadratic inequality on \(|\bar{u}| \geq 0\) admits the upper bound
	\begin{align*}
		|\bar{u}| \leq \frac{  2{\left(
		\|V\|_{L^2}^2\|\tilde{u}\|_{L^2}^2 + \langle Vu, u \rangle_{L^2} \int_{M}V
		\right)}^{1/2} } {\int_{M}V }.
	\end{align*}
	Hence,
	\begin{align*}
	{\bar{u}}^2|M| &\leq \frac{4|M|}{{\left(\int_{M}V\right)}^2} \left(
			\|V\|_{L^2}^2 \|\tilde{u}\|_{L^2}^2 + \langle Vu, u \rangle_{L^2} \int_{M}V
	\right) \\
	&\leq \frac{4|M|}{{\left(\int_{M}V\right)}^2} \left(
	C_{P}^2\|V\|_{L^2}^2{\|\nabla u\|}_{L^2}^2 + \langle Vu, u \rangle_{L^2} \int_{M}V
	\right) 
	\end{align*}
	Substituting and grouping terms,
	\begin{align*}
		{\|u\|}^2_{L^2} \leq C_{P}^2\left(1 + \frac{4|M| \|V\|_{L^2}^2}{{\left(\int_{M}V\right)}^2} \right) {\|\nabla^2 u\|}_{L^2}^2
		+ \frac{4|M|}{\int_{M}V} \langle Vu, u \rangle_{L^2}.
	\end{align*}
	The lemma follows from this inequality.
\end{proof}
\begin{lemma}\label{lem:l1-l2-inv-op-bound}
	For any \(\epsilon > 0\) there is a constant \(K > 0\) such that for any valid configuration of vortices and antivortices and for any \(h \in H^2\) with \(\|h\|_{H^2} < \epsilon \),
	\[
	{\|L^{-1}\|} \leq K,
	\] 
	where the operator \(L\) is defined by equation~\eqref{eq:op-l1}.
\end{lemma}
\begin{proof}
	Let \(w \in L^{2}\) be a unit vector and let \(u \in H^2\) be such that,
	\begin{equation}
		-\nabla^2 u + 2qf'(h+v) u = w \label{eq:nabla-u-w}
	\end{equation}
	Let \(V = 2qf'(h+v)\), multiply~\eqref{eq:nabla-u-w} by \(u\) and integrate by parts to obtain the following,
	\begin{align*}
		\|u\|_{L^2}^2 &\leq \alpha (\|\nabla u\|_{L^2}^2 + \langle Vu, u\rangle_{L^2}) \nonumber\\
		 &= \alpha \langle u, w\rangle_{L^2} \nonumber \\
		 &\leq \alpha \|u\|_{L^2}.
	\end{align*}
	In the first inequality we used the constant \(\alpha\) given by  Lemma~\ref{lem:u-l2-bound};  
	hence, \(\|u\|_{L^2} \leq \alpha\). By 
	Lemma~\ref{lem:pot-bound}, there is a constant \(C > 0\) independent of \(h\) and vortex-antivortex position such that \(\langle f'(h+v) , 1 \rangle_{L^2} \geq C\), whence, 
	\[ \alpha \leq \max \left( 
	C_{P}^2 \left(1 + 4|\Sigma|^2\frac{{\|f'(x)\|}_{L^{\infty}}^2}{C^2}\right),
	\frac{2|\Sigma|}{qC}
	\right). \]
	Therefore, there is a constant \(C_{2} > 0\), independent of \(h\) and vortex-antivortex position, such that \(\|u\|_{L^2} < C_2\). 
	By Schauder's estimate, there is a constant \(C_{S}\) such that,
	\begin{align}
		\|u\|_{H^2} &\leq C_{S} (\|\nabla^2 u\|_{L^2} + \|u\|_{L^2}) \nonumber \\
		&\leq C_{S} (1 + (2q \|f'(x)\|_{L^{\infty}} + 1) \|u\|_{L^2})\nonumber \\
		&\leq C_{S} (1 + (2q \|f'(x)\|_{L^{\infty}} + 1) C_{2}). \label{eq:uh2-upper-bound-1}
	\end{align}
	In the second inequality, we used~\eqref{eq:nabla-u-w} and the triangle inequality. Let \(K\) be the right hand side of~\eqref{eq:uh2-upper-bound-1}, we conclude that \(\|L^{-1}w\|_{H^2} = \|u\|_{H^2} \leq K\) whenever \(\|w\|_{L^2} = 1\), proving the lemma.
\end{proof}

\begin{lemma}
	\label{lem:h0-bound}
	For \(\kappa = 0\) and any valid configuration \((\mathbf{p}, \mathbf{q}) \in \Sigma^{k_{+}} \times \Sigma^{ k_{-}}\), the solution $h$ of~\eqref{eq:h-n-cal} is 
	uniformly bounded in $H^2$.
\end{lemma}

\begin{proof}
	Let \(\bar{h} = |\Sigma|^{-1}\int_{\Sigma} h \) 
	be the mean of $h$ on $\Sigma$ and let \(\tilde{h} = h - \bar{h}\); by
	Schauder's estimates there is a positive constant \(C_{S}\) independent
	of \(h\) such that by~\eqref{eq:h-n-cal},
	\begin{align*}
		\|\tilde{h}\|_{H^2} &\leq C_{S} \|2q f(h+v) + A\|_{L^2} \\
		&\leq C_{S}(2q \|f(x)\|_{L^{\infty}} + |A|) {|\Sigma|}^{1/2}.	
	\end{align*}
	Therefore, the set of functions 
	$\tilde{h}$ is bounded in $H^2$ and by Sobolev's
	embedding also in $C^0(\Sigma)$. We claim  that the set of mean values 
	is also bounded. Assume towards a
	contradiction the existence of sequences $v_n$, $\tilde{h}_n$, $\bar{h}_n$ such that 
	$|\bar{h}_n| \to \infty$. Suppose  $\bar{h}_n \to \infty$, and let 
	$(\mathbf{p}_n, \mathbf{q}_n)  \in {\Sigma}^{k_{+}} \times \Sigma^{k_{-}}$ be the points
	defining \(v_n = v_{\mathbf{p}_n} - v_{\mathbf{q}_n}\). Since $\Sigma$ is compact, we can assume 
	convergence $(\mathbf{p}_n, \mathbf{q}_n) \to (\mathbf{p}_{*}, \mathbf{q}_{*}) \in
	\Sigma^{k_{+}} \times \Sigma^{k_{-}}$. We have pointwise convergence $v_n
	\to  v_{*} = v_{\mathbf{p}_{*}} - v_{\mathbf{q}_{*}}$, except possibly
	at a finite set of points where a coordinate of $\mathbf{p}_{*}$
	coincides with some coordinate of $\mathbf{q}_{*}$. Since $\{\tilde
	h_n\}$ is \(C^{0}\)-bounded,
	\begin{equation*}
		f(v_n + \tilde h_n + \bar{h}_n) \to f_{\max}, \qquad \text{pointwise a.e.}
	\end{equation*}
	By the dominated convergence theorem,
	\begin{equation*}
		\int_{\Sigma}f(v_n + \tilde h_n + \bar{h}_n) \to f_{\max}|\Sigma|.
	\end{equation*}
	On the other hand, if we apply the divergence theorem to~\eqref{eq:h-n-cal} with \(\kappa = 0\),
	\begin{equation*}
		2q \int_{\Sigma} f (v_n + \tilde h_n + c_n) =
		- A |\Sigma|.
	\end{equation*}
	Whence, \(A = -2qf_{\max} \)  and this contradicts~\eqref{eq:bradlow-cond}. If $c_n\to -\infty$ the same
	argument yields another contradiction.
\end{proof}

\begin{proposition}\label{prop:minimal-constant}
	For any \(q > 0\) such that condition~\eqref{eq:bradlow-cond} holds, there is
	a constant \( \kappa_{*} > 0 \) such that for any valid pair \( (\mathbf{p}, \mathbf{q}) \in \Sigma^{k_{+}} \times \Sigma^{k_{-}} \) and any \(\kappa\) with \( |\kappa| < \kappa_{*}\), system~\eqref{eq:h-n-cal}--\eqref{eq:n-cal} admits a solution.
\end{proposition}
\begin{proof}
	By Lemma~\ref{lem:h0-bound}, there is a constant \(R > 0\) such that any solution \((h, 0)\) of system~\eqref{eq:h-n-cal}--\eqref{eq:n-cal} at \(\kappa = 0\) is contained in the ball \(B_{R}(\mathbf{0}) \subset H^2 \times H^2\). By Lemma~\ref{lem:l1-l2-inv-op-bound}, there is a constant \(K > 0\) such that if \((h, \mathcal{N}) \in B_{R}(\mathbf{0})\), the operator \(L^{-1}\)  defined by~\eqref{eq:op-l1} has norm bounded by \(K\). Let 
	\begin{align*}
		\varepsilon &= \min \left(R, \frac{1}{16q\|f''(x)\|_{L^{\infty}}K}\right), \\
		\kappa_{0} &= \min\left(1, \frac{1}{8qK (2 + q(\|f'(x)\|_{L^{\infty}} + 1))}\right),
	\end{align*}
	Assume \(|\kappa| \leq \kappa_{0}\) and \((h, \mathcal{N}) \in B_{\varepsilon}(\mathbf{0})\), by equations~\eqref{eq:l-l-inv-p-norm}--\eqref{eq:l-inv-p-norm},
	\begin{align*}
		\|(L\oplus L)^{-1}P\| &\leq 2K (\kappa_{0} q (2 + q(\|f'(x)\|_{L^{\infty}} +  1)) + 2q\|f''(x)\|_{L^{\infty}}\varepsilon) \\
		&\leq \frac{1}{2}.
	\end{align*}
	Hence, \(dT_{\kappa}\) is invertible and \(\|dT_{\kappa}^{-1}\| \leq 4K\) by~\eqref{eq:dtk-inv},~\eqref{eq:neumann-series}. Let $\chi_{\kappa} = (h_{\kappa}, \mathcal{N}_{\kappa})$ be the curve of solutions to
	system~\eqref{eq:h-n-cal}--\eqref{eq:n-cal} guaranteed to exist by
	Proposition~\ref{prop:small-extension}. This curve can be extended whenever $d_{\chi_{\kappa}}{T}_{\kappa}$ is invertible. This is the case if $|{\kappa}| < \kappa_{0}$  and \(\chi_{\kappa} \in B_{\varepsilon}(\mathbf{0})\). Since \(T(\kappa, \chi_{\kappa}) = 0\), implicit derivation shows that 
	\[ \|\dot{\chi}_{\kappa}\|_{H^2\times H^2} = \|d_{\chi_{\kappa}}T_{\kappa}^{-1} \partial_{\kappa}T\|_{H^2\times H^2} \leq 4K \|\partial_{\kappa}T\|_{L^2\times L^2}. \]
	Since \(\partial_{\kappa} T = (2q\mathcal{N}, q^2 (2\kappa \mathcal{N} + f(h + v)))\), we find the bound
	\[ \|\dot{\chi}_{\kappa}\|_{H^2\times H^2} \leq 4K (2q\varepsilon (1 + q) + q^2 \|f(x)\|_{L^{\infty}}|\Sigma|^{1/2}) \]
	valid as long as \(|\kappa| \leq \kappa_{0}\) and \(\chi_{\kappa} \in B_{\varepsilon}(\mathbf{0})\). Let \(\kappa^{*} \in (0, \infty]\) be the maximal extension of \(\chi_{\kappa}\) such that it is well defined for \(\kappa \in [0, \kappa^{*})\). Either \(\kappa^{*} > \kappa_{0}\) or there is some \(\kappa < \kappa^{*}\) such that \(\|\dot{\chi}_{\kappa}\|_{H^2\times H^2} > \varepsilon/ \kappa^{*}\); hence,
	\begin{align*}
		\kappa^{*} > \min \left(
		\kappa_{0}, \frac{\varepsilon}{4K (2q (1 + q) \varepsilon + q^2 \|f(x)\|_{L^{\infty}}|\Sigma|^{1/2})}
		\right).
	\end{align*}
	Define \(\kappa_{*}^{+}\) as the right side of this inequality, then system~\eqref{eq:h-n-cal}--\eqref{eq:n-cal} admits a solution for \(\kappa \in [0, \kappa_{*}^{+})\), regardless of vortex-antivortex position. 
	For \(\kappa \leq 0\), a similar argument proves the existence of another constant \(\kappa_{*}^{-} > 0\) such that the proposition holds for \(\kappa \in (-\kappa_{*}^{-}, 0]\). Take \(\kappa_{*} = \min(\kappa_{*}^{+}, \kappa_{*}^{-})\).
\end{proof}

\section{\large{Continuous extensions of the O(3) Sigma model}}\label{sec:cont-extension}

Proposition~\ref{prop:small-extension} shows that for small \(\kappa\) there is a solution of the elliptic system~\eqref{eq:h-n-cal}--\eqref{eq:n-cal} that varies smoothly with \(\kappa\). The aim of this section is to show that this small extension can be extended continuously. Given a solution \((h,\mathcal{N})\) to the elliptic system, it will be convenient to define the scalar field \(N = -\kappa \mathcal{N}\) and the equivalent system 
\begin{equation}
	\begin{aligned}
		\nabla^2 h & = -2q \left(
		N - f({h + v})
		\right) + A,	\\ %\label{eq:elliptic-sys} \\
		\nabla^2 N & = \kappa^2 q^2 \left(
		N - f({h + v})\right) + 2 q f'({h + v}) N. 
	\end{aligned} \label{eq:elliptic-sys}
\end{equation}
The algebra property of \(H^k\), \(k \geq 2\), and elliptic regularity show that any weak solution of~\eqref{eq:elliptic-sys} is a strong solution of class \(C^{\infty}\). As a consequence of the maximum principle, 
\begin{equation}
	f_{\min} \leq N \leq f_{\max }.\label{eq:n-bounds}
\end{equation}
If system~\eqref{eq:elliptic-sys} admits a solution for \(A \neq 0\), \(\kappa \) must be bounded, since by the divergence theorem,
\begin{align}
	-2q \int_{\Sigma } (N - f({h + v})) + A |\Sigma|                     & = 0, \label{eq:div-h} \\
	k^2q^2\int_{\Sigma}(N - f({h + v})) + 2q \int_{\Sigma} f'({h+v})N & = 0, \nonumber
\end{align}
whence,
\begin{align}
	\frac{A}{2q} |\Sigma|  =
	\int_{\Sigma } (N - f({h + v}))                                                        
	= \frac{-2}{\kappa^2q} \int_{\Sigma} f'({h+v})N. \label{eq:k2-integral}
\end{align}
\eqref{eq:n-bounds} together with~\eqref{eq:k2-integral} implies,
\begin{align*}
	\kappa^2 = \frac{-4}{A |\Sigma|} \int_{\Sigma} f'({h+v})N, 
	\leq \frac{4}{|A|} \|f'(x)\|_{{L}^{\infty}} \|f(x)\|_{L^{\infty}}.
\end{align*}
For \(q \),~\eqref{eq:div-h} implies
\begin{align*}
	|A||\Sigma| & \leq 2q \int_{\Sigma}
	\left|
	N - f({h + v})
	\right| \leq 2q (f_{\max} - f_{\min}) |\Sigma|.
\end{align*}
Hence, 
\begin{equation*}
	q \geq \frac{|A|}{2 (f_{\max} - f_{\min}) }.
\end{equation*}
Given an arbitrary \(\lambda \in \mathbb{R}^{+} \), consider the operators \( \Psi = -\nabla^2 + \lambda \) and \(\Phi: \mathbb{R} \times H^{k} \times H^{k} \to H^{k} \times H^{k}\) with components
\begin{align}
	\Phi_{1}(\kappa, h, N) & = -2q(N - f({h+v})) + A + \lambda h, \label{eq:phi_1}                       \\
	\Phi_{2}(\kappa, h, N) & = \kappa^2q^2 (N - f({h+v})) + 2q f'({h+v}) N + \lambda N. \label{eq:phi_2}
\end{align}

\begin{lemma}\label{lem:f-fp-hk-bound}
	Let \(S \subset H^{k}\), \(k \ge 2\), be a bounded
	subset. For the function \(f(x)\) defined
	in~\eqref{eq:f-cp1-model}, there exists a constant \(K > 0\) such
	that for any \(h \in S\),  
	\[
	\|f(h+v)\|_{H^{k}} \le K, \qquad {\|f'(h + v)\|}_{H^k} \le K.
	\]
\end{lemma}
\begin{proof}
	Recall that \(v = v_{+} - v_{-}\), where the functions \(v_{\pm}\) are
	defined in~\eqref{eq:vpm} and that each function \(e^{v_{\pm}}\) is smooth
	and vanishes only at vortex or antivortex positions.
	For any \(k \ge 2\), the Sobolev space \(H^{k}\) is an algebra.
	Equation~\eqref{eq:f-fp-hv} shows that it suffices to show that the family
	\(
	\{ (e^{h} e^{v_{+}} + e^{v_{-}})^{-1} \mid h \in S \}
	\)
	is \(H^{k}\)-bounded.
	By Sobolev's embedding, any \(h \in H^{k}\) is continuous and there
	exists a constant \(C_{\mathrm{Sob}}\) such that
	\({\|h\|}_{C^{0}} \le C_{\mathrm{Sob}} {\|h\|}_{H^{k}}\).
	Since \(S\) is bounded in \(H^{k}\), it is uniformly bounded
	in \(C^{0}\). In particular, there exists a constant \(\delta > 0\) such that
	\(
	e^{h} e^{v_{+}} + e^{v_{-}} \ge \delta 
	\) for all \(h \in S\). 
	Define \(g : [\delta, \infty) \to \mathbb{R}\) by \(g(x) = 1/x\) and extend
	\(g\) smoothly to a function \(\tilde{g} : \mathbb{R} \to \mathbb{R}\) with 
	bounded derivatives such that \(\tilde{g}(0) = 0\). 
	Since \(\|e^{h}\|_{H^{k}} \le C \exp(\|h\|_{H^{k}})\) for some constant independent of \(h\), there exists another constant
	\(K\) such that
	\[
	{\| e^{h} e^{v_{+}} + e^{v_{-}} \|}_{H^{k}} \le K,
	\qquad \text{for all } h \in S.
	\]
	Let \(u = e^{h} e^{v_{+}} + e^{v_{-}}\), 
	\(\tilde{g}\) satisfies the conditions for the composition result
	\cite[Cor., p.~281]{runstSobolevSpacesFractional2011}. Hence, there
	is a constant \(K_{1}\), independent of \(u\), such that
	\[
	{\left\| \frac{1}{u} \right\|}_{H^{k}}
	= {\| \tilde{g}(u) \|}_{H^{k}}
	\le
	K_{1}\bigl( {\|u\|}_{H^{k}} + {\|u\|}_{H^{k}} {\|u\|}_{C^{0}}^{k-1} \bigr).
	\]
	Since \({\|u\|}_{H^{k}} \le K\) and
	\({\|u\|}_{C^{0}} \le C_{\mathrm{Sob}} {\|u\|}_{H^{k}}\), the right-hand side is
	uniformly bounded for all \(h \in S\).
	This proves the lemma.
\end{proof}

For the next proposition, recall an operator between Banach spaces is compact if it sends
bounded sets into precompact sets. 
\begin{proposition}\label{prop:leray-schauder-op-compact}
	For any \(k \geq 2\) and any \(\kappa \in \mathbb{R}\), the operator 
	\[ {(\Psi \oplus \Psi)}^{-1}\Phi(\kappa, \cdot, \cdot): H^{k} \times H^{k} \to H^{k} \times H^{k}  \]
	is compact.
\end{proposition}
\begin{proof}
	By Lemma~\ref{lem:f-fp-hk-bound} and the algebra property of \(H^k\), 
	\(\Phi\) sends bounded sets in \(H^{k}\times H^{k}\) into bounded sets
	in \(H^{k}\times H^{k}\). On the other hand, 
	\[ (\Psi \oplus \Psi)^{-1}: H^{k}\times H^{k} \to H^{k+2}\times H^{k+2} \]
	is a bounded operator; 
	hence, \((\Psi \oplus
	\Psi)^{-1}\Phi \) is also bounded. The inclusion \(H^{k+2}
	\to H^{k} \) is compact by the Rellich-Kondrachov
        theorem. Therefore,
	\((\Psi \oplus \Psi)^{-1}\Phi\) is compact as a map \(H^{k}\times H^{k} \to
	H^{k}\times H^{k}\).
\end{proof}

\begin{lemma}\label{lem:uniform-cont-balls}
	Let \(k \geq 2\), for any ball \(B \subset H^{k}\times H^{k} \) and
	any \(\epsilon >
	0\), there is a \(\delta > 0 \) such that for any \((h, N) \in B \),
	if \(|\kappa_2 - \kappa_1| < \delta \), then 
	\[ \|(\Psi \oplus
	\Psi )^{-1}\Phi(\kappa_2, h, N) - (\Psi \oplus \Psi)^{-1}\Phi(\kappa_1, h, N)\|_{H^k}
	< \epsilon.  \]
\end{lemma}
\noindent It is said that the operator \((\Psi\oplus \Psi)^{-1}\Phi\) is continuous
in \(\kappa\) uniformly with respect to \((h, N)\) in balls.
\begin{proof}
	\((\Psi \oplus \Psi)^{-1} \) is bounded; hence, it suffices to
	show that \(\Phi\) satisfies the condition of the lemma. For fixed
	\((h, N)\), \(\Phi_1\) given by equation~\eqref{eq:phi_1} is independent of \(\kappa\),
	\[ \|\Phi_1(\kappa_2, h, N) - \Phi_1(\kappa_1, h, N)\|_{H^{k}} = 0. \]
	 Since \( B \) is a bounded set, by Lemma~\ref{lem:f-fp-hk-bound} there is a constant \(C > 0\) such that,
	\begin{align} 
		\|\Phi_2(\kappa_2, h, N) - \Phi_2(\kappa_1, h, N)\|_{H^{k}} & = |k_2^2-
		k_1^2|q^2 \| N - f({h+v})\|_{H^{k}} \nonumber
		\\
		                                                            & \leq C
		|\kappa_2^2 - \kappa_1^2|.\label{eq:phi2-phi1}
	\end{align}
	Equation~\eqref{eq:phi2-phi1} shows that \(\Phi_2\) is uniformly
	continuous in \(\kappa\) with respect to \((h, N) \in B \).
\end{proof}
Proposition~\ref{prop:leray-schauder-op-compact} shows that for any \(\kappa > 0 \) and any bounded set \(\Omega
\subset H^{k}\times H^{k} \), \(k \geq 2 \),  such that 
\(0 \not\in (I - (\Psi \oplus \Psi)^{-1}\Phi(\kappa, \cdot, \cdot))\left( \partial \Omega \right)\), 
the Leray-Schauder degree, \(\deg(I - (\Psi \oplus \Psi)^{-1}\Phi(\kappa,
\cdot, \cdot), \Omega, 0)\), is a well defined homotopical invariant
with respect to deformations of the parameter \(\kappa\), provided the
solution at \(\kappa\) does not intersect the boundary \(\partial\Omega\). 
Notice that
\((I - (\Psi \oplus \Psi)^{-1}\Phi(\kappa,\cdot, \cdot))(h, N) = 0\) if and only if \((h,
N) \) is a solution of~\eqref{eq:elliptic-sys}. By
Proposition~\ref{prop:small-extension}, there are an \(\epsilon > 0\)
and \(\kappa_0 > 0\), such that if \((h_0, 0) \) is the solution of system~\eqref{eq:elliptic-sys} at \(\kappa = 0\), then for any \(\kappa \)
such that \(|\kappa| < \kappa_0\), there is exactly one solution,
continuously depending on \(\kappa\), of this system in the ball
\(B_{\epsilon}((h_0,0)) \subset H^{k}\times H^{k}\).

\begin{lemma}\label{lem:deg-compact-deformation}
	For small \(\epsilon > 0\), \(|\deg(I - (\Psi \oplus \Psi)^{-1}\Phi(0,
	\cdot, \cdot), B_{\epsilon}(h_0, 0), 0)| = 1 \).
\end{lemma}
\begin{proof}
	Recall \(\Psi = -\nabla^2 + \lambda \) and that \(\Phi\) is given by
	equations~\eqref{eq:phi_1}--\eqref{eq:phi_2}. Let \(F = (\Psi \oplus
	\Psi )^{-1}\Phi(0, \cdot, \cdot)\). For \(\epsilon\)
	small, the Leray-Schauder theorem~\cite[Thm.~1.10]{chipotHandbookDifferentialEquations2004} asserts
	\[
	\deg(I - F, B_{\epsilon}(h_0, 0), 0) 
	= \pm 1,
	\] 
	provided \(I - F':
	H^{k}\times H^{k} \to H^{k}\times H^{k}\) is injective. Let \((\delta h,
	\delta N) \in Ker(I - F')\), then \(F'(\delta h, \delta N) = (\delta
	h, \delta N) \). Linearity of \((\Psi \oplus \Psi)^{-1} \) implies 
	\[ (\Psi
	\delta h, \Psi \delta N) = \Phi'(0, h_0, 0)(\delta h, \delta
	N).  \]
	Equivalently,
	\begin{align}
		\nabla^2 \delta h & = -2q \delta N + 2q f'({h + v}) \delta
		h, \label{eq:1}                                            \\
		\nabla^2 \delta N & = 2q f'({h + v}) \delta N.\label{eq:2}
	\end{align}
	Elliptic regularity and the algebra property of \(H^{k}\) imply \(\delta h, \delta N \in H^{k+2}\); moreover,
	the operator \(-\nabla^2 + 2qf'({h+v}): H^{k + 2} \to H^{k} \)
        is injective. By~\eqref{eq:2}, \(\delta N = 0\);
        then, \(\delta h = 0 \) by~\eqref{eq:1}.
\end{proof}
Lemmas~\ref{lem:uniform-cont-balls}
and~\ref{lem:deg-compact-deformation} imply the existence of two connected
and closed sets~\cite[Thm.~3.3]{chipotHandbookDifferentialEquations2004}, \(C^+ \subset [0, \infty) \times H^{k} \times H^{k} \) and
	\(C^- \subset (-\infty, 0] \times H^{k} \times H^{k}\) such that each
\((\kappa, h, N) \in C^{\pm} \) is a solution of
system~\eqref{eq:elliptic-sys} and \((0, h_0, 0)
\in C^+ \cap C^-\). Moreover, either 
\[
C^{\pm} \text{ is unbounded} \quad \text{or} \quad
C^{\pm}
\cap \left(\left\{ 0 \right\} \times (H^{k}\times H^{k} \setminus
B_{\epsilon}(h_0, 0))\right) \neq \emptyset.
\]
The second alternative is not possible because for \(\kappa = 0 \)
the unique solution to~\eqref{eq:elliptic-sys} is the given \((h_{0}, 0)\); whence, each of
these sets is unbounded. In the following propositions, we restrict without loss of
generality to \(C^+\) and study the asymptotic behaviour of diverging sequences when the vortex and antivortex numbers are different. Given an integrable function \(h\), denote by \(\bar{h} = {|\Sigma|}^{-1}
\int_{\Sigma} h\) its mean value.

\begin{proposition}\label{prop:maxwell-lim-non-top-h2}
	If \(A \neq 0 \) and \(\{(\kappa_i, h_i, N_i)\} \subset \mathbb{R} \times
	H^2 \times H^2 \), is a divergent sequence
	of solutions to System~\eqref{eq:elliptic-sys}, up to a
	sub-sequence the following limit holds,
	\begin{equation}
		\lim_{i \to \infty} |\kappa_i| + \|h_i - \bar{h}_i\|_{H^2} + \left\|N_{i} - \frac{A}{2q} - \ell\right\|_{H^2} = 0, \label{eq:lim-h2-unbalance}
	\end{equation}
	where
	\begin{equation}
		\ell = \begin{cases}
			f_{\min}, & \text{if } A > 0, \\
			f_{\max}, & \text{if } A < 0.
		\end{cases}
	\end{equation}
	Moreover,
	\begin{equation}
		\lim_{i \to \infty } \bar{h}_{i} = \begin{cases}
			-\infty, & \text{if } A > 0, \\
			\infty,  & \text{if } A < 0.
		\end{cases}
	\end{equation}
\end{proposition}
\begin{proof}
	\(\kappa\) is bounded for \(A \neq 0 \).  Let
        \(\tilde{h}_{i} = h_{i} - \bar{h}_{i} \), \(f(x)\) and \(f'(x)\) are bounded functions, as well as \(N\) by~\eqref{eq:n-bounds}; hence, the right hand side of~\eqref{eq:elliptic-sys} is \(L^2\) bounded. By the standard elliptic estimates, \(\{\tilde{h}_{i}\}\) and \(\{N_{i}\}\) are \(H^2\) bounded. We claim that \(\{\bar{h}_{i}\}\) is unbounded, otherwise, \(\{h_{i}\}\) is a bounded set in \(H^2\). Hence, up to a subsequence, 
        \(\lim_{i \to \infty } \kappa_i = \kappa_0 \) for some \(\kappa_{0}\) and \(\lim_{i \to \infty} \bar{h}_{i} = \pm\infty\). 
  		Passing to a subsequence,
        \(\tilde{h}_{i} \to \tilde{h}_{0}\) and \(N_{i}
        \to N_{0}\) weakly in \(H^1\) and strongly in \(L^2\). Notice that
        the sets \(\{\tilde{h}_{i}\}\) and \(\{N_{i}\}\) are
        \(C^0\)-bounded by Sobolev's embedding. Therefore, except for vortex-antivortex points, \(\lim_{i \to \infty} f({\tilde{h}_{i} +
          \bar{h}_{i} + v}) = \ell \), where
          \[ 
          \ell = \begin{cases}
          	f_{\max} &\text{if } \bar{h}_{i} \to \infty, \\
          	f_{\min} &\text{if } \bar{h}_{i} \to -\infty.
          \end{cases}
           \]
          Weak convergence in \(H^1\)
        together with strong convergence in \(L^2\) imply \((\tilde{h}_{0},
        N_{0})\) is a solution of the elliptic PDEs,
	\begin{align}
		\nabla^2 \tilde{h}_{0} & = -2q \left(N_{0} - \ell\right) + A,\label{eq:h0}        \\
		\nabla^2 N_{0}         & = k_{0}^2 q^2 \left( N_{0} - \ell \right).\label{eq:n0}
	\end{align}
	A bootstrap argument shows \((\tilde{h}_{0}, N_{0})\) is a 
        strong \(C^2\) solution. By the standard
        elliptic estimates, the convergence \(N_{i} \to N_{0} \) is
        strong in \(H^2 \) and by Sobolev's embedding theorem, it is
        also strong in \(C^{0}\). From equation~\eqref{eq:k2-integral} we find,
	\begin{equation}
		4 \int_{\Sigma} f'\left(\tilde{h}_{i} + \bar{h}_{i}+v \right) N_{i} dV = - \kappa_{i}^2 A |\Sigma|. \label{eq:int-fp-e-n}
	\end{equation}
	Taking limits, the left side converges to 0 by the dominated convergence theorem; whence,
        \(\kappa_{0} = 0\) and \(N_{0}\) is constant by
        equation~\eqref{eq:n0}. The maximum principle 
        applied to equation~\eqref{eq:h0} implies,
	\begin{equation}
		N_{0} = \frac{A}{2q} + \ell\label{eq:n0-alpha}
	\end{equation}
	and \(\tilde{h}_{0} \equiv 0\), since \(\tilde{h}_{0} \) is a
        constant function of zero mean. Notice that \(f' \) is
        always positive; if \(A > 0 \), the right side of
        equation~\eqref{eq:int-fp-e-n} is negative, whence, the
        uniform convergence \(N_{i} \to N_{0} \) implies \(N_{0} < 0
        \). Therefore, \(\ell = f_{\min} < 0 \) by~\eqref{eq:n0-alpha} and  \(\bar{h}_{i} \to
        -\infty\). A similar argument shows that if \(A < 0 \),
        \(\ell = f_{\max}\) and \(\bar{h}_{i} \to \infty \).
\end{proof}

Let \(\tilde{h}_{i} = h_{i} - \bar{h}_{i} \), Proposition~\ref{prop:maxwell-lim-non-top-h2} and Sobolev's embedding imply,
\[ \lim_{i \to \infty} \|\tilde{h}_{i}\|_{C^0} + \left\|N_{i} - \frac{A}{2q} - \ell \right\|_{C^0} = 0. \]
Hence, \(\{\tilde{h}_{i}\}\) and \(\{N_{i}\}\) also converge in \(L^p\) for any \(p > 2\). The right hand side of~\eqref{eq:elliptic-sys} also converges to 0 in \(L^p\) as \(i \to \infty\). Schauder's estimate and Sobolev's embedding imply respectively the following limits, 
\begin{align}
	\lim_{i \to \infty} \|\tilde{h}_{i}\|_{W^{2,p}} + \left\|N_{i} - \frac{A}{2q} - \ell \right\|_{W^{2,p}} &= 0, \nonumber \\
\lim_{i \to \infty} \|\tilde{h}_{i}\|_{C^1} + \left\|N_{i} - \frac{A}{2q} - \ell \right\|_{C^1} &= 0. \label{eq:lim-c1-unbalance}	
\end{align}

In the pure vortex and pure antivortex cases,  Proposition~\ref{prop:maxwell-lim-non-top-h2} can be strengthen.

\begin{proposition}
	If \(k_{+} = 0\) or \(k_{-} = 0\), for any divergent sequence of solutions to system~\eqref{eq:elliptic-sys},
	\[ \{(\kappa_i, h_i, N_i)\} \subset \mathbb{R} \times
	H^2 \times H^2, \]
	up to a sub-sequence, the following limit holds for all \(k \geq 0\),
	\begin{equation}
		\lim_{i \to \infty} \|h_i - \bar{h}_i\|_{C^k} + \left\|N_{i} - \frac{A}{2q} - \ell\right\|_{C^k} = 0, 
	\end{equation}
	where \(\ell\) is defined as in Proposition~\ref{prop:maxwell-lim-non-top-h2}.
\end{proposition}
\begin{proof}
	By Proposition~\ref{prop:maxwell-lim-non-top-h2} we know limit~\eqref{eq:lim-h2-unbalance} is valid; moreover, the proof of the proposition shows \(\{\tilde{h}_{i}\}\) and \(\{N_{i}\}\) are bounded in \(H^2\). If \(k_{+} = 0\), \(A < 0\) and \(\bar{h}_{i} \to \infty\) as \(i \to \infty\). Hence, \(\{f(h_{i} + v)\}\) and \(\{f'(h_{i} + v)\}\) are bounded in \(H^2\), whence, the right hand side of~\eqref{eq:elliptic-sys} is also bounded in \(H^2\). A recursive argument using the standard elliptic estimates shows \(\{\tilde{h}_{i}\}\) and \(N_{i}\) are bounded in \(H^k\) for all \(k \geq 2\).
	We claim the following limit holds for \(k \geq 2\),
	\begin{equation}
		\lim_{i \to \infty} \|h_i - \bar{h}_i\|_{H^k} + \left\|N_{i} - \frac{A}{2q} - \ell\right\|_{H^k} = 0. \label{eq:lim-ind}
	\end{equation}
	This limit and Sobolev's embedding prove the proposition in the case \(k_{+} = 0\). The case \(k_{-} = 0\) is similar. To prove the limit, we use induction on \(k\). The first step was proved in Proposition~\ref{prop:maxwell-lim-non-top-h2}. Assume~\eqref{eq:lim-ind} is valid for some \(k \geq 2\). Notice that as \(i \to \infty\),
	\begin{gather*}
		\|f'(\tilde{h}_{i} + \bar{h}_{i} - v_{-})\|_{H^k} = e^{-\bar{h}_{i}}\left\|\frac{2e^{\tilde{h}_{i}}e^{v_{-}}}{(e^{\tilde{h}_{i}} + e^{-\bar{h}_{i}}e^{v_{-}})^2}\right\|_{H^{k}} \leq C e^{-\bar{h}_{i}} \to 0, \\
		\|f(\tilde{h}_{i} + \bar{h}_{i} - v_{-}) - f_{\max}\|_{H^k} = e^{-\bar{h}_{i}}\left\|\frac{2 e^{v_{-}}}{e^{\tilde{h}_{i}} + e^{-\bar{h}_{i}}e^{v_{-}}}\right\|_{H^{k}} \leq C e^{-\bar{h}_{i}} \to 0.
	\end{gather*}
	These limits and~\eqref{eq:lim-ind} imply the right hand side of~\eqref{eq:elliptic-sys} converges to 0 in \(H^k \times H^k\); then, by Schauder's estimates,~\eqref{eq:lim-ind} is also valid for \(H^{k + 2}\). This completes the induction.
\end{proof}
 
\subsection*{Proof of Theorem~\ref{thm:existence}}

Given a solution \((\phi_{\kappa}, \mathcal{A}_{\kappa}, \mathcal{N}_{\kappa})\) to the Bogomol'nyi equations, we can find a solution \((u_{\kappa}, \mathcal{N}_{\kappa})\) to the governing elliptic  system~\eqref{eq:field-theory-system}. Conversely, any solution to the system determines a gauge equivalence class of solutions to the Bogomol'nyi equations. 

\emph{Proof of 1.} Given \(k \geq 0\), let \((h_{\kappa}, \mathcal{N}_{\kappa})\) be the family of solutions to~\eqref{eq:h-n-cal}--\eqref{eq:n-cal} given by Proposition~\ref{prop:small-extension}. This family varies continuously, in fact smoothly, with \(\kappa\). Define 
\begin{equation}
(\phi_{3})_{\kappa} = -1 + \frac{2e^{v_{-}}}{e^{h_{\kappa}}e^{v_{+}} + e^{v_{-}}}, \label{eq:phi3-as-fn-h}
\end{equation}
hence, \(((\phi_{3})_{\kappa}, \mathcal{N}_{\kappa})\) also varies continuously with \(\kappa\). Therefore, for any \(\epsilon > 0\) we can find \(\delta > 0\) such that, if \(|\kappa| < \delta\), 
\[\|(\phi_{3})_{\kappa} - (\phi_{3})_{0}\|_{H^{k+2}} + \|\mathcal{N}_{\kappa}\|_{H^{k+2}} < \epsilon.\]
Equation~\eqref{eq:top-limit-k-0} is consequence of Sobolev's embedding. 

\emph{Proof of 2.} By Proposition~\ref{prop:minimal-constant}, we can find a constant \(\kappa_{*}\), independent of the core set, such that if \(|\kappa| < \kappa_{*}\), there is a solution \((h_{\kappa}, \mathcal{N}_{\kappa})\) to~\eqref{eq:h-n-cal}--\eqref{eq:n-cal}. This proposition and Lemma~\ref{lem:h0-bound} show that there is an \(R > 0\) such that if \(|\kappa| <  \kappa_{*} \), we can find 
	\((h_{\kappa}^{1}, \mathcal{N}_{\kappa}^{1}) \in B_{R}(\mathbf{0}) \subset H^2 \times H^2\). 		
	We claim that making \(\kappa_{*} \) smaller if necessary, there is a  second solution \((h_{\kappa}, \mathcal{N}_{\kappa})\) such that 
	\begin{equation}
	\|h_{\kappa}\|_{H^2}^2 + \|\mathcal{N}_{\kappa}\|^2_{H^2} > R^2. \label{eq:unbounded-sequence-2}
	\end{equation}
	Assume otherwise, then we may find a sequence \(\kappa_{i} \to 0 \) such that for any \(\kappa_i\) and for any solution \((h_{\kappa_i}, \mathcal{N}_{\kappa_{i}})\)  to~~\eqref{eq:h-n-cal}--\eqref{eq:n-cal}, 
	\[ \|h_{\kappa_i}\|^{2}_{H^2} + \|\mathcal{N}_{\kappa_{i}}\|_{H^2}^2 \leq R^2. \]
	Assume without loss of generality \(k_{i} > 0\) for all \(i\). 
	Since \(C^{+}\) is connected and unbounded, we may find a second sequence \( \{(\kappa_{j}, h_{j}, \mathcal{N}_{j})\} \) such that \(\kappa_{j} \to 0\) and 
	\[ \|h_{j}\|^{2}_{H^2} + \|\mathcal{N}_{j}\|_{H^2}^2 = 2R^2. \]
	Up to a subsequence, \(\|h_{j} - h_{0}\|^2_{H^2} + \|N_{j}\|^2_{H^2} \to 0\) as \(j \to \infty\), where \((h_{0}, 0)\) is the 
	\(O(3)\)-Sigma model solution to~\eqref{eq:h-n-cal}--\eqref{eq:n-cal} with the same vortex-antivortex data. A contradiction. 

\emph{Proof of 3.} Assume \(k_{+} > k_{-}\). The proof of part 2 shows that there is a constant \(R > 0\) and a family \(\mathcal{F} = \{(h_{\kappa}, \mathcal{N}_{\kappa}) \mid \kappa \in (0, \kappa_{*})\}\) such that~\eqref{eq:unbounded-sequence-2} holds. In order to prove  limit~\eqref{eq:non-top-limit-unbalance}, it is sufficient to prove that whenever \(\kappa_{i} \to 0\), there is a subsequence, also denoted \(\kappa_{i}\), such that~\eqref{eq:non-top-limit-unbalance}  holds for \((h_{i}, \mathcal{N}_{i})\). If \(\kappa_{i} \to 0\), then 
\[ 
\|h_{\kappa_{i}}\|^2_{H^2} + \|\mathcal{N}_{\kappa_i} \|_{H^2}^2 \to \infty,
\]
since otherwise, there is a convergent subsequence \((h_{\kappa_j}, \mathcal{N}_{\kappa_{j}}) \to (h_{0}, 0)\) in contradiction to~\eqref{eq:unbounded-sequence-2}. 
Let \(U \subset \Sigma\) be an open neighbourhood of the core set, in \(\Sigma \setminus \bar{U}\) the functions \(e^{v_{\pm}}\) are smooth positive. 
Let \(\bar{h}_{i} = |\Sigma|^{-1}\int_{\Sigma} h_{\kappa_{i}}\), \(\tilde{h}_{i} = h_{\kappa_{i}} - \bar{h}_{i}\), by Proposition~\ref{prop:maxwell-lim-non-top-h2}, \(\bar{h}_{i} \to -\infty \) and by~\eqref{eq:lim-c1-unbalance}, \(\tilde{h}_{i}\) is uniformly bounded in \(C^1\); hence, also \(e^{\tilde{h}_{i}}\) and by the product rule \(e^{\tilde{h}_{i}}e^{v_{+}}\). There is a constant \(C > 0\) such that
\begin{equation*}
	\left\| \frac{1}{e^{\bar{h}_{i}} e^{\tilde{h}_{i}} e^{v+} + e^{v_{-}}} \right\|_{C^0(\Sigma \setminus \bar{U})} \leq \|e^{-v_{-}}\|_{{C^0(\Sigma \setminus \bar{U})}} < C,
\end{equation*} 
\begin{multline*}
	\left\| \nabla ({e^{\bar{h}_{i}} e^{\tilde{h}_{i}} e^{v+} + e^{v_{-}}})^{-1} \right\|_{C^0(\Sigma \setminus \bar{U})} \\
	\leq  \|e^{-2v_{-}}\|_{C^0(\Sigma\setminus \bar{U})} \left\| 
	e^{\bar{h}_{i}}\nabla (e^{\tilde{h}_{i}} e^{v_{+}}) 
	+ \nabla e^{v_{-}}
	\right\|_{C^{0}(\Sigma \setminus \bar{U})} < C.
\end{multline*}
It follows that, 
\begin{align}
	\|(\phi_{3})_{\kappa_{i}} - 1\|_{C^{1}(\Sigma\setminus \bar{U})} =
	\left\|\frac{2 e^{\bar{h}_{i}} e^{\tilde{h}_{i}} e^{v_{+}}}{e^{\bar{h}_{i}} e^{\tilde{h}_{i}} e^{v+} + e^{v_{-}}}\right\|_{C^1(\Sigma\setminus \bar{U})} \to 0, \qquad \text{as } i \to \infty. \label{eq:lim-phi3-1}
\end{align}
Recall in equation~\eqref{eq:lim-c1-unbalance}, \(N_{i} = -\kappa_{i}\mathcal{N}_{\kappa_{i}}\), \(A = 4\pi (k_{+} - k_{-}) {|\Sigma|}^{-1}\) and \(\ell = -1 + \tau\), hence, 
\[ \lim_{i \to \infty }\left\|\kappa_{i}\mathcal{N}_{\kappa_{i}} + \frac{2\pi (k_{+} - k_{-})}{q |\Sigma|} - 1 + \tau\right\|_{C^1} = 0. \]
This proves the limit for \(k_{+} > k_{-}\). The case \(k_{+} < k_{-}\) is analogous.

\section{Large continuous extensions}

In the previous section we showed that small extensions of the regular elliptic system~\eqref{eq:elliptic-sys} extend continuously in \(\mathbb{R} \times H^k \times H^2\) to unbounded sets of solutions while \(\kappa\) remains bounded. 
In this section we show that if \(k_{+} = k_{-}\), the deformation parameter \(\kappa\) extends to arbitrarily large values and prove Theorem~\ref{thm:limit-infty}. 

\begin{lemma}\label{lem:cm-bounded}
	If \(A = 0\), for any \(M > 0\) and \(k \geq 0\), the set
	\begin{equation*}
		C_{M} = \{(h_{\kappa}, N_{\kappa}) \mid |\kappa| \leq M \text{ and } (h_{\kappa}, N_{\kappa}) \text{ solves~\eqref{eq:elliptic-sys} with param. } \kappa \}
	\end{equation*}
	is bounded in \(H^{k+2} \times H^{k+2}\).
\end{lemma}
\begin{proof}
	Given \(h \in H^{k + 2} \), let \(\bar{h}\) be the mean value and \(\tilde{h} = h - \bar{h} \).
	We prove the lemma by induction on the Sobolev
        space. For bounded \(\kappa\), the right side
        of~\eqref{eq:elliptic-sys} is bounded in \(L^2\) since \(f(x)\), \(f'(x)\) and \(N\) are uniformly bounded. 
        Elliptic regularity implies \(\tilde{h}\) and
        \(N\) are \(H^2\)-bounded; in particular, \(\tilde{h}\) is
        also \(C^{0}\)-bounded by Sobolev's embedding. Assume towards
        a contradiction the existence of a sequence \((h_{n}, N_{n})\) of
        solutions to~\eqref{eq:elliptic-sys} such that \(|\bar{h}_{n}|
        \to \infty\). Passing to a subsequence, we may assume
        the limits \(\tilde{h}_{n} \to
        \tilde{h}_{0} \), \(N_{n} \to N_{0}\) weakly in \(H^1\) and strongly in \(L^2\) and \(\kappa_{n} \to \kappa_{0}\), 
        \(\bar{h}_{n} \to \pm \infty\). If \(\bar{h}_{n} \to \infty\), except for antivortex positions, 
        \(f(h_{n}+v) \to f_{\max}\) because \(\tilde{h}_{n}\) is
        uniformly bounded. Hence, the same limit holds in \(L^2\). 
	In this case, \((\tilde{h}_{0}, N_{0})\) is a weak solution of the system
	\begin{align}
		\nabla^2 \tilde{h}_{0} & = -2q(N_{0} - f_{\max}),\label{eq:h0-2}              \\
		\nabla^2 N_{0}         & = \kappa_{0}^2q^2 (N_{0} - f_{\max}).\label{eq:n0-2}
	\end{align}
	Elliptic regularity implies \((\tilde{h}_{0}, N_{0})\) is also a
        strong solution of class \(C^2\). We claim 
        \(N_{0} = f_{\max}\). If \(\kappa_{0} \neq 0 \), this is due
        to~\eqref{eq:n0-2} and the maximum principle. 
        If \(\kappa_{0} = 0 \), then
        \(N_{0}\) is constant. Equation~\eqref{eq:h0-2} implies the
        constant is again \(f_{\max}\). Whence, \(N_{n}\) converges
        uniformly to a non-zero constant, whereas \(f'(h_{n} + v)\) is
        positive except for a finite number of zeroes. This
        contradicts~\eqref{eq:int-fp-e-n}, since \(A = 0\). If \(\bar{h}_{n} \to
        -\infty \) the proof is analogous. Therefore, the set of
        solutions to system~\eqref{eq:elliptic-sys} is bounded in
        \(H^2\). To conclude the induction, assume for some \(k \geq 2
        \) there is a constant
        \(K \) such that \(\|h\|_{H^{k}} + \|N\|_{H^{k}} \leq K \) for 
        \((h, N) \in C_{M}\). Hence, the right side
        of~\eqref{eq:elliptic-sys} is also bounded by a constant in
        \(H^k\). By Schauder's estimates, there is another constant,
        still denoted \(K\), such that \(\|h\|_{H^{k+2}} +
        \|N\|_{H^{k+2}} \leq C\). Since \(H^{k+2}\) embeds continuously into
        \(H^{k+1}\), this concludes the induction.
\end{proof}

\begin{proposition}\label{prop:existence-sols-a-0}
	If \(k_{+} = k_{-}\), for any \(q > 0\), \(\kappa \in \mathbb{R}\) and any  prescribed vortex-antivortex positions, there is a solution \((u, \mathcal{N})\) to system~\eqref{eq:field-theory-system}. Except for positions of vortices and antivortices, this solution is smooth.
\end{proposition}
\begin{proof}
	\(A = 4\pi(k_{+}-k_{-}){|\Sigma|}^{-1} = 0\), hence,   condition~\eqref{eq:bradlow-cond} is trivially satisfied. It follows that for any vortex-antivortex data, there is a solution \((h, 0)\) to~\eqref{eq:elliptic-sys} at \(\kappa = 0\) that extends continuously to unbounded, connected sets of solutions \(C^{\pm} \subset {\mathbb{R}}^{\pm}_{0} \times H^2 \times H^2 \).
	The canonical projection \(\Pi: {\mathbb{R}}^{\pm}_{0} \times H^2 \times H^2 \to {\mathbb{R}}^{\pm}_{0} \) is continuous; hence, \(\Pi(C^{\pm}) \subset {\mathbb{R}}^{\pm}_{0}\) is a connected subset that contains \(0\). By lemma~\ref{lem:cm-bounded}, there is a sequence \(\{\kappa_{n}\} \subset \Pi(C^{\pm}) \) such that \(\kappa_{n} \to \pm \infty \) as \(n \to \infty\). This shows \(\Pi(C^{\pm}) = {\mathbb{R}}^{\pm}_{0}\); whence, for any \(\kappa \neq 0\) there is an \(H^2\) solution \((h_{\kappa}, N_{\kappa})\) to~\eqref{eq:elliptic-sys} which is also smooth by elliptic regularity. Recall \(u_{\kappa} = h_{\kappa} + v\),  \(\mathcal{N}_{\kappa} = - \kappa^{-1} N_{\kappa}\) is a solution to~\eqref{eq:field-theory-system} and \(v\) is smooth away of vortex-antivortex positions. 
\end{proof}

\subsection{Asymptotics at large \(\kappa\)}\label{sec:asympt-large}

Proposition~\ref{prop:existence-sols-a-0} shows that if \(A = 0\), for any \(\kappa\) and \(q\) there is a solution to system~\eqref{eq:elliptic-sys}. 
Let \(u = h + v\), it will be convenient to define \(w = N - f(u)\).  A short computation shows 
\begin{align*}
	\nabla^2 f(u) = -2q f'(u)w + f''(u) |\nabla u|^2.
\end{align*}
Subtracting this equation from the equation for \(N\) in~\eqref{eq:elliptic-sys}, 
\begin{equation}
	\nabla^2 w = (\kappa^2q^2 + 2qf'(u))w + 2qf'(u) N  - f''(u){|\nabla u|}^2. \label{eq:elliptic-w}
\end{equation}

\begin{lemma}\label{lem:fpp-du-bound}
	For any bounded set \(S \subset H^k\), \(k \geq 2\), the set 
	\[ \{
	f''(h+v)|\nabla v_{\pm}|^2 \mid h \in S
	\} \]
	is also bounded in \(H^k\).
\end{lemma}
\begin{proof}
	In local coordinates \((x,y)\) centred at vortex (antivortex) position, 
	functions \(v_{\pm}\) satisfy the approximation 
	\[ v_{\pm} = \log (x^2 + y^2) + \tilde{v} \]
	for some smooth function \(\tilde{v}\). 
	This shows that \(e^{v_{\pm}}\) and \(e^{v_{+}} e^{v_{-}} {|\nabla v_{\pm}|}^2\) are smooth functions. Notice that \( f''(x) = {2e^x (1 - e^x)}{{(e^x + 1)}^{-3}}, \) then,
	\[ f''(h + v){|\nabla v_{\pm}|}^{2} = \frac{2 e^h (e^{v_{-}} - e^{h}e^{v_{+}})}{{(e^{h}e^{v_{+}} + e^{v_{-}})}^3} e^{v_{+}}e^{v_{-}} {|\nabla v_{\pm}|}^2. \]
	If \(C > 0\) is a constant such that \(\|h\|_{H^k} < C\) for \(h \in S\), 
	then there are constants \(M, K > 0\) such that \(\|e^h\|_{H^k} \leq M\) and \(\|{(e^{h}e^{v_{+}} + e^{v_{-}})}^{-1}\|_{H^k} < K\). Hence, 
	\begin{align*}
		\left\|f''(h+v) {|\nabla v_{\pm}|}^2 \right\|_{H^k} \leq 2MK^3 (\|e^{v_{-}}\|_{H^k} + M\|e^{v_{+}}\|_{H^k}) {\left\|e^{v_{+}}e^{v_{-}} {|v_{\pm}|}^2\right\|}_{H^k}.
	\end{align*}
\end{proof}
We know the functions \(h - \bar{h}\) are always \(H^2\)-bounded. Given a sequence of solutions to~\eqref{eq:elliptic-sys}, \((\kappa_{i}, h_{i}, N_{i})\), up to a subsequence, either \(\bar{h}_{i} \to c_{*}\) or \(\bar{h}_{i} \to \pm \infty\). In the first case, \(\{h_{i}\}\) is \(H^2\)-bounded and by Lemma~\ref{lem:fpp-du-bound} also \(\{f''(u_{i}) {|\nabla u_{i}|}^2 \}\). The second case is addressed in the following Lemma.
\begin{lemma}\label{lem:thinks-are-complicated-at-infty}
	Let \(\{h_{i}\}\) be a \(W^{2,p}\)-sequence, \(p > 2\), such that \(\{h_{i} - \bar{h}_{i}\}\) is bounded and \(\bar{h}_{i} \to \pm\infty\) and let \(u_{i} = h_{i} + v\),  then \(\{f''(u_i) {|\nabla u_{i}|}^2 \}\) is bounded in \(L^1\). 
\end{lemma}
\begin{proof}
	\(f''(x) = 2 (e^x{(e^x + 1)}^{-2}) (1 - 2 e^x {(e^x + 1)}^{-1})\), whence,
	\[ |f''(x)| \leq \frac{6 e^x}{(e^x + 1)^2}. \]
	It is sufficient to prove the lemma for \(g(x) = e^x{(e^x + 1)}^{-2}\). Assume without loss of generality \(\bar{h}_{i} \to \infty \) and let \((\mathbf{p}, \mathbf{q}) \in \Sigma^{k_{+}} \times \Sigma^{k_{-}}\) be the positions of the vortices and antivortices. Let  \(U_{n} \subset \Sigma \) be an small open neighbourhood about \(p_{n}\), \(n = 1, \ldots, k_{+}\) and let \(V = \Sigma \setminus \cup_{n} \bar{U}_{n}\) be an open set containing only antivortices. Notice that \(v = v_{+} - v_{-}\), where \(v_{+} \), \(v_{-}\) defined on~\eqref{eq:vpm} satisfy \(e^{v_{\pm}}\) is smooth and \(e^{v_{+}} > 0\) on \(V\). Let \(\tilde{h}_{i} = h_{i} - \bar{h}_{i}\) and let \(C > 0\) be a constant such that \(\|\tilde{h}_{i}\|_{H^2} \leq C\). 
	By Sobolev's embedding, 
	\(\{\tilde{h}_{i}\}\) is also bounded in \(C^1(\Sigma)\); hence, there is a constant \(K\) such that \(\|e^{-\tilde{h}_{i}}e^{-v_{+}}\|_{C^0(V)} \leq K\) and \({\|\nabla \tilde{h}_{i}\|}_{C^1(\Sigma)} < K\). Since \(g(x) \leq e^{-x} \), on \(V\) we have,
	\begin{align*}
		g(u_{i}) {|\nabla u_{i}|}^2 &\leq e^{-\bar{h}_{i}} e^{-\tilde{h}_{i}} e^{-v_{+}}e^{v_{-}} |\nabla (\tilde{h}_{i} + v)|^2 \\
		&\leq e^{-\bar{h}_{i}} K e^{v_{-}}(K^2 + 2 K |\nabla v| + |\nabla v|^2).
	\end{align*}
	Functions \(e^{v_-}|\nabla v|\) and \(e^{v_{-}} |\nabla v|^2 \) are continuous in \(\bar{V}\), whence, there is a constant \(M > 0\) such that on \(V\),
	\[ g(u_{i}) {|\nabla u_{i}|}^2 \leq e^{-\bar{h}_{i}} M. \]
	This shows \(g(u_{i}) {|\nabla u_{i}|}^2 \to 0\) as \(i \to \infty\) in \(L^1(V)\). For any \(U_{n}\), there are constants \(K_{1}\), \(K_{2} > 0\) such that \(K_{1} \leq e^{\tilde{h}_{i} - v_{-}} \leq K_{2}\), hence,
	\begin{align}
		g(u_{i}){|\nabla u_{i}|}^2 &\leq \frac{e^{\bar{h}_{i}}e^{v_{+}} K_{2}}{{(e^{\bar{h}_{i}}e^{v_{+}}K_{1} + 1)}^2} (K^2 + 2K |\nabla v| + |\nabla v|^2). \label{eq:gui-nablaui-bound}
	\end{align} 
	Assume \(U_{n}\) is sufficiently small for the existence of a chart \(\varphi_{n}: U_{n} \to D_{1}(0)\), where \(D_{1}(0) \subset \mathbb{R}^2 \) is the unit euclidean disk. Assume further \(\varphi_{n}(p_{n}) = 0\) and that the metric is conformal to the euclidean metric, \(\Omega (dx_{1}^2 + dx_{2}^2)\). Let us denote by \((r, \theta)\) polar coordinates in \(D_{1}(0)\), then 
	\[ v_{+}\circ \varphi_{n}^{-1} = 2 d \log r + \tilde{v}_{+}, \]
	where \(d\) is the degree of \(p_{n}\) and \(\tilde{v}_{+}\) is smooth. Hence, we can find another constants, also denoted \(K_{1}\) and \(K\), such that by~\eqref{eq:gui-nablaui-bound}, 
	\[ \int_{U_{n}} g(u_{i}){|\nabla u_{i}|}^2 \leq 2\pi \int_{0}^{1} \frac{e^{\bar{h}_{i}} K r^{2d-1} dr}{{(e^{\bar{h}_{i}} r^{2d} K_{1} + 1)}^2} = 
	\frac{\pi K}{d K_{1}}\left(
	1 - \frac{1}{e^{\bar{h}_{i}}K_{1} + 1}
	\right).
	\]
	Therefore, \(\int_{\Sigma} g(u_{i}) {|\nabla u_{i}|}^2 \) is bounded. The case \(\bar{h}_{i} \to -\infty\) is analogous.
\end{proof}
\begin{proposition}\label{lem:fpp-bounds}
	There is a constant \(C > 0\) such that for any solution \((h, N)\) of~\eqref{eq:elliptic-sys},
	\begin{equation}
		\kappa q \|N - f(h+v)\|_{L^2} \leq C. \label{eq:w-bound}
	\end{equation}
	If we restrict to a bounded set \(S \subset H^2\), then  
	\begin{equation}
		\kappa^2 q^2 \|N - f(h+v)\|_{L^2} \leq C, \quad \forall h \in S. \label{eq:w-bound-s}
	\end{equation}
\end{proposition}
\begin{proof}
	Let \(L = -\nabla^2 + \kappa^2q^2 + 2qf'(u) \). Recall \(w = N - f(h + v)\) and notice \(\|w\|_{L^\infty} \leq 2\|f(x)\|_{L^{\infty}}\). 
	Multiply~\eqref{eq:elliptic-w} by \(w\) and rearrange terms to obtain the following inequality,
	\begin{align}
		\kappa^2 q^2 \|w\|^2_{L^2} &\leq \langle Lw, w \rangle_{L^2} \nonumber \\
		&\leq |\langle f''(u)|\nabla u|^2, w\rangle_{L^2}| + 2q|\langle f'(u) N, w \rangle_{L^2}|. \label{eq:k2q2w-bound}
	\end{align}
	If we restrict to the bounded set \(S\), by lemmas~\ref{lem:f-fp-hk-bound} and~\ref{lem:fpp-du-bound} there is a constant \(C > 0\) such that \(\|f'(u)\|_{H^2} \leq C\) and \(\|f''(u)|\nabla u|^2\|_{H^2} \leq C\). \(N\) is uniformly bounded by \(\|f(x)\|_{L^{\infty}}\). Applying Cauchy-Schwartz to~\eqref{eq:k2q2w-bound}, we find that there is a constant, also denoted \(C\), such that 
	\[ \kappa^2q^2 \|w\|^2_{L^2} \leq C \|w\|_{L^2};\]
	therefore,~\eqref{eq:w-bound-s} holds in this case. For the general case, notice that 
	\begin{align*}
		|\langle f'(u)N, w \rangle_{L^2}| &\leq 2\|f'(x)\|_{L^\infty}\|f(x)\|_{L^{\infty}}^2, \\
		|\langle f''(u){|\nabla u|}^2, w\rangle_{L^2}| &\leq 2\|f(x)\|_{L^{\infty}} \|f''(u)|\nabla u|^2\|_{L^1}.
	\end{align*}
	For any solution \((h, N)\) to~\eqref{eq:elliptic-sys}, \(h \in W^{2,p}\), \(p \geq 2\), and \(h - \bar{h}\) is \(C^0\)-bounded. By lemma~\ref{lem:thinks-are-complicated-at-infty}, \(\|f''(u)|\nabla u|^2\|_{L^1}\) is bounded. 
	Therefore, the right side of~\eqref{eq:k2q2w-bound} is bounded by Lemma~\ref{lem:thinks-are-complicated-at-infty}; whence,~\eqref{eq:w-bound} holds.
\end{proof}

\begin{proposition}\label{prop:lim-infty}
	Let~\((\kappa_{i}, h_{i}, N_{i})\) be a sequence of solutions to~\eqref{eq:elliptic-sys} such that \(\kappa_{i} \to \pm \infty\). Up to a subsequence, one of the following alternatives holds as \(i \to \infty\),
	\begin{enumerate}
		\item \(\bar{h}_{i} \to \infty\) and \(\|h_{i} - \bar{h}_{i}\|_{C^1} + \|N_{i} - f_{\max}\|_{L^p} \to 0\), \(p \geq 2\).
		\item \(\bar{h}_{i} \to -\infty\) and \(\|h_{i} - \bar{h}_{i}\|_{C^1} + \|N_{i} - f_{\min}\|_{L^p} \to 0\), \(p \geq 2\).
		\item There is a constant \(c\) such that 
		\[ \int_{\Sigma} f'(c + v)f(c + v) = 0 \]
		and for any \(k \geq 0\), \(\|h_{i} - c\|_{C^k} + \|N_{i} - f(c + v)\|_{C^k} \to 0 \).
	\end{enumerate}
\end{proposition}
\begin{proof}
	Let \(u_i = h_{i} + v\). 
	By proposition~\ref{lem:fpp-bounds}, \(\|N_{i} - f(u_{i})\|_{L^2} \to 0\) as \(i \to \infty\). \(\|N_{i} - f(u_{i})\|_{L^\infty}\) is bounded; by interpolation,
	\begin{equation}
		\lim_{i\to\infty} \|N_{i} - f(u_{i})\|_{L^p} = 0, \qquad p \geq 2. \label{eq:n-f-lim}
	\end{equation}
	By elliptic regularity and the equation for \(h\) in \eqref{eq:elliptic-sys}, \(\|h_{i} - \bar{h}_{i}\|_{W^{2,p}} \to 0\). Taking any \(p > 2\), Sobolev's embedding implies \(\|h_{i} - \bar{h}_{i}\|_{C^1} \to 0 \) as \(i \to \infty\). 
	Suppose \(\{\bar{h}_{i}\}\) unbounded and assume without loss of generality \(\bar{h}_{i} \to \infty\) as \(i \to \infty\). Since \(\{h_{i} - \bar{h}_{i}\}\) is \(C^0\)-bounded, \(\|f(u_{i}) - f_{\max}\|_{L^p} \to 0\) by the dominated convergence theorem. The triangle inequality and equation~\eqref{eq:n-f-lim} imply \(\lim_{i\to\infty} \|N_{i} - f_{\max}\|_{L^p} = 0\), \(p \geq 2\). The case \(\bar{h}_{i} \to -\infty\) is analogous. If \(\{\bar{h}_{i}\}\) is bounded, up to a subsequence,  \(\lim_{i\to\infty} \bar{h}_{i} = c\) for some constant \(c\) to be determined. We claim that 
	\begin{equation}
		\lim_{i\to \infty} \|h_{i} - c\|_{H^k} = 0, \quad  \lim_{i\to \infty} \|N_{i} - f(c + v)\|_{H^k} = 0, \qquad k \geq 2. \label{eq:induction-lim}
	\end{equation}
	This claim together with Sobolev's embedding imply the third alternative. We proceed by induction on Sobolev space \(H^k\). 
	Since \(\bar{h}_{i} \to c\), \eqref{eq:n-f-lim} and Schauder's estimates applied to the equation for \(h_{i}\) in~\eqref{eq:elliptic-sys} imply 
	\[ \lim_{i\to \infty} \|h_{i} - c\|_{H^2} = 0. \]
	In particular, \(\{h_{i}\}\) is \(H^2\)-bounded. 
	Continuity of the function \(h \mapsto f(h +v)\) as a map \(H^k \to H^k\), \(k \geq 2\), imply \(\lim_{i\to \infty} \|f(u_{i}) - f(c + v)\|_{H^2} = 0\). By Proposition~\ref{lem:fpp-bounds}, \(\{\kappa_{i}^2q^2 \|N_{i} -f(u_{i})\|_{L^2}\}\) is bounded and by Lemma~\ref{lem:f-fp-hk-bound}, \(\{f'(u_{i})\}\) is \(H^2\)-bounded. Hence, \(\{N_{i}\}\) is \(H^2\)-bounded by Schauder's estimates and equation~\eqref{eq:elliptic-sys}. 
	Let~\(\epsilon = (\kappa q)^{-2}\), recall \(w = N - f(u)\) in~\eqref{eq:elliptic-w}. This equation can be rewritten as
	\[ (-\epsilon\nabla^2 + 1 + 2q \epsilon f'(u_{i}))w_{i} = \epsilon (f''(u_{i})|\nabla u_{i}|^2 - 2q f'(u_{i})N_{i}). \]
	By \cite[Lem.~2.4]{ricciardiAsymptoticsNonlinearElliptic2003}, there is a constant \(C_{k} > 0\), independent of \(h_{i}\) and \(N_{i}\), and an  \(\epsilon_{0} > 0\) such that, for \(\epsilon \leq \epsilon_{0}\), 
	\begin{equation}
		\|w_{i}\|_{H^k} \leq C_{k}  \epsilon_{i} \|f''(u_{i})|\nabla u_{i}|^2 - 2q f'(u_{i})N_{i}\|_{H^k}, \qquad k \geq 2.  \label{eq:ricciardi-trick}
	\end{equation}
	Whence, \(\lim_{i\to \infty}\|N_{i} - f(u_{i})\|_{H^2} = 0\). It follows that
	\[ \lim_{i \to \infty} \|N_{i} - f(c + v)\|_{H^2} = 0. \]
	Assume~\eqref{eq:induction-lim} holds for some \(k \geq 2\). Then, \(\{h_{i}\}\) and \(\{N_{i}\}\) are bounded in \(H^k\). This implies \(f''(u_{i})|\nabla u_{i}|^2\) and \(f'(u_{i})\) are \(H^k\) bounded by Proposition~\ref{lem:fpp-bounds} and Lemma~\ref{lem:f-fp-hk-bound}. Equation~\eqref{eq:ricciardi-trick} implies \(\lim_{i \to \infty} \|N_{i} - f(u_{i})\|_{H^k} = 0\); hence, 
	\[ \lim_{i\to \infty} \|h_{i} - c\|_{H^{k+1}} \leq  \lim_{i\to \infty} \|h_{i} - c\|_{H^{k+2}} = 0 \]
	by Shauder's estimates applied to~\eqref{eq:elliptic-sys}. By continuity,
	\begin{equation}
		\lim_{i \to \infty} \|f(u_{i}) - f(c + v)\|_{H^{k+1}} = 0.  \label{eq:lim-fui}
	\end{equation}
	Note that the equation for \(N\) in~\eqref{eq:elliptic-sys} can be rearranged as 
	\[ (-\epsilon \nabla^2 + 1 + 2q \epsilon f'(u_{i})) N_{i} = f(u_{i}), \]
	whence, we can suppose that for \(\epsilon_{0}\) defined previously,
	\[ \|N_{i}\|_{H^k} \leq C_{k} \|f(u_{i})\|_{H^k}, \qquad k \geq 2, \]
	provided \(\epsilon_{i} \leq \epsilon_{0}\). Therefore, \(N_{i}\) is bounded in \(H^{k+1}\) as well as \(f'(u_{i}) N_{i}\) since \(H^{k+1}\) is an algebra. Equation~\eqref{eq:ricciardi-trick} implies \(\lim_{i\to \infty} \|N_{i} - f(u_{i})\|_{H^{k+1}} = 0\). This limit, together with~\eqref{eq:lim-fui} imply 
	\[ \lim_{i \to \infty} \|N_{i} - f(c + v)\|_{H^{k+1}} = 0. \]
	This concludes the induction. Since \(A = 0\), equation~\eqref{eq:int-fp-e-n} and the uniform convergence \(h_{i} \to c\), \(N_{i} \to f(c + v)\) imply
	\[ \int_{\Sigma} f'(c + v)f(c + v) = \lim_{i\to\infty} \int_{\Sigma} f'(u_{i})N_i = 0. \]
\end{proof}

\subsection*{Proof of Theorem~\ref{thm:limit-infty}}

The first part of the theorem is a direct consequence of Proposition~\ref{prop:existence-sols-a-0}. Given a sequence \((\kappa_{i}, \phi_{i}, \mathcal{A}_{i}, \mathcal{N}_{i}) \) of solutions to the Bogomol'nyi equations, it determines a sequence \((\kappa_{i}, h_{i}, N_{i}) \) of solutions to the elliptic system~\ref{eq:elliptic-sys}. Assume without loss of generality \(\kappa_{i} \to \infty\) as \(i \to \infty\). Let \(U\) be an open neighbourhood of the core set. Recall \(N = -\kappa \mathcal{N}\), \(f_{\max} = 1 + \tau\) and \(f_{\min} = -1 + \tau \). 
The limits in Theorem~\ref{thm:limit-infty} are restatements of the limits in Proposition~\ref{prop:lim-infty}. 
If \(\bar{h}_{i} \to \pm \infty \), 
\begin{equation*}
	\lim_{i \to \infty}\|(\phi_{3})_{i} \pm 1\|_{C^1(\Sigma \setminus \bar{U})} = 0
\end{equation*}
as in the proof of Theorem~\ref{thm:existence}, (e.g. equation~\eqref{eq:lim-phi3-1}). 
 If \(\bar{h}_{i}\) is bounded, the third alternative in Proposition~\ref{prop:lim-infty} implies for any \(k \geq 0\),
\[ \|h_{i} - c\|_{H^{k+2}} \to 0, \qquad \text{as } i \to \infty. \]
Equation~\eqref{eq:phi3-as-fn-h} shows that as a function \(H^{k+2} \to H^{k+2}\), \(h \mapsto \phi_{3}(h)\) is continuous. By Sobolev's embedding,
\[ \lim_{i \to \infty}\|(\phi_{3})_{i} - \phi_{3}(c)\|_{C^k} \leq C_{Sob} \lim_{i \to \infty}\|(\phi_{3})_{i} - \phi_{3}(c)\|_{H^{k+2}} = 0. \]

\section{Symmetric deformations on the
  sphere}\label{sec:symm-deform-sphere}

In this section we study Chern-Simons deformations on the sphere. 
Assume vortices are located at the north pole
of the round sphere, whereas antivortices are located at the south pole. 
Choose trivialisations $\phi_{\pm}: U_{\pm} \to \mathbb{S}^2$, where 
$U_{\pm} = \mathbb{S}^2 \setminus \{(0, 0, \mp 1)\}$, and stereographic projection charts from the south or north pole respectively,
$\varphi_{\pm}:U_{\pm} \to \mathbb{C}$. These projections are
related by a gauge transformation, which by spherical symmetry is,
\begin{equation*}
	\varphi_+ = \frac{e^{in\theta}}{\varphi_-}, \qquad x \in U_+\cap U_-.
\end{equation*}
If the connection is represented locally by $a_{\pm} \in \Omega^1(U_{\pm})$, then 
\begin{equation*}
	a_+ = a_- + n d\theta, \qquad x \in U_+ \cap U_-.
\end{equation*}

Stereographic coordinates will be denoted
 $x_{\pm} = r_{\pm} e^{i\theta_{\pm}}$. Hence, in $U_+\cap U_-$, 
 $x_+x_- = 1$ and $\theta_+ =-\theta_-$. We choose the ansatz,
\begin{align*}
	\varphi_{\pm} & = f_{\pm}(r_{\pm})e^{i k_{\pm}\theta_{\pm}}, &
	a_{\pm}       & = a_{\pm}(r_{\pm})\,d\theta_{\pm},
\end{align*}
which is justified by rotational symmetry. Compatibility of the fields then requires $n = k_+ -
	k_-$. The Bogomol'nyi equations reduce to a system of ODEs,
\begin{align*}
	f_{\pm}'  & = \frac{1}{r} (k_{\pm} \mp a_{\pm})\, f_{\pm},                                                                 \\
	a_{\pm}'  & = r \Omega(r) B_{\pm},                                                                                         \\
	N_{\pm}'' & = -\Omega(r) \left(\kappa B_{\pm} - \frac{4 f_{\pm}^2N_{\pm}}{(1 + f_{\pm}^2)^2}\right) - \frac{1}{r}N'_{\pm},
\end{align*}
where,
\begin{equation*}
	\Omega(r) = \frac{4R^2}{(1 + r^2)^2}, \qquad
	B_{\pm}   = - \left(\kappa N_{\pm} + \tau \pm 1 \mp \frac{2}{1 + f_{\pm}^2}\right).
\end{equation*}

We solved the Bogomol'nyi equations in the punctured disk
$\mathbb{D}_1(0)\setminus\{0\}$ adding the compatibility conditions,
\begin{align*}
	f_+(1)f_-(1)    & = 1,         &
	a_+(1) + a_-(1) & = n,         &
	N_+(1)          & = N_-(1),    &
	N_+'(1)         & = - N_-'(1),
\end{align*}
together with the lowest order approximation to the fields at $r = 0$,
\begin{align*}
	f_{\pm} & = q_{\pm}r^{k_{\pm}} + \mathcal{O}\left(r^{k_{\pm} + 1}\right),
	\\
	B_{\pm} & =
	\begin{cases}
		-\left(\kappa p_{\pm} + \tau \pm 1 \mp \frac{2}{1 + q_{\pm}^2}\right) + \mathcal{O}(r), &
		k_{\pm} = 0,                                                                              \\
		-(\kappa p_{\pm} + \tau \pm 1 \mp 2) + \mathcal{O}(r),                                  &
		k_{\pm} \neq 0,                                                                           \\
	\end{cases}
	\\
	a_{\pm} & = 2 B_{\pm}R^2r^2 + \mathcal{O}(r^3),
	\\
	N_{\pm} & =
	\begin{cases}
		p_{\pm} + \left(
		-\kappa B_{\pm} + \frac{4q_{\pm}^2p_{\pm}}{(1 + q_{\pm}^2)^2}
		\right)\,R^2r^2 + \mathcal{O}(r^3),
		 & k_{\pm} = 0,    \\
		p_{\pm} -
		\kappa B_{\pm}
		\, R^2r^2 + \mathcal{O}(r^3),
		 & k_{\pm} \neq 0.
	\end{cases}
\end{align*}

To find the initial stable solution, we used a shooting method in
the interval $[\delta, 1]$ for a small value $\delta > 0$. Given
initial conditions for the parameters, $Z =
(q_+, q_-, p_+, p_-)$, we solved the Bogomol'nyi
equations on \(\mathbb{D}_{1}(0)\) and defined a map $M: \mathbb{R}^4 \to \mathbb{R}^4$,
\begin{align*}
	Z \mapsto (f_+(1)f_-(1) - 1, a_+(1) + a_-(1) - n,
	N_+(1) - N_-(1), N'_+(1) + N'_-(1)),
\end{align*}
whose zero determines suitable initial conditions. In order to obtain a family of Chern-Simons deformations, we
applied the pseudo-arclength continuation method described in
\cite{flood2018chern} for deformations of the Abelian-Higgs model. Given initial data $(\kappa_0, Z_0) \in \mathbb{R}^5$, we sought a nearby point $(\kappa, Z)$ such that,
\begin{equation*}
	\dot Z_0 \cdot (Z - Z_0) + \dot \kappa_0\,(\kappa - \kappa_0) =
	\delta s,
\end{equation*}
for a small positive constant $\delta s$. We restricted to
$\kappa > 0$ and solved the Bogomol'nyi equations for \(\kappa_{+} = 2, \kappa_{-} = 0\) in a sphere of radius \(R = 2\) and for \(\kappa_{+} = 1, \kappa_{-} = 1\) in a sphere of radius \(R = 1.5\). 

Figure~\ref{fig:sols20} shows field evolution in the $(2, 0)$ case. The fields converge according to Theorem~\ref{thm:existence} as expected. We remark that the theorem claims \(dN_{i} = \kappa_{i} d\mathcal{N}_{i} \to 0\) as \(i \to \infty\); however, we cannot rule out the possibility that \(d\mathcal{N}_{i}\) diverges.  
Since \(e = d\mathcal{N}\) by~\eqref{eq:e}, data suggests  \(d\mathcal{N}_{i}\) is bounded.  
Figure~\ref{fig:sols11} shows field evolution in the $(1, 1)$ case. 
Data suggests that in this 
case for any $\kappa > 0$ there is exactly one solution to the Bogomol'nyi equations. We reduced the radius of the sphere to keep pseudo arclength continuation convergent. In this case, Theorem~\ref{thm:limit-infty} gives three possible limits for the fields. In the sphere, Green functions are of the form 
\[ G(x,y) = \frac{1}{4\pi}\log (1 - x\cdot y) + c \]
for some constant \(c \in \mathbb{R}\). A direct calculation shows that condition~\eqref{eq:c-cond} holds for \(c = 0\). 
Data supports that the convergence in this case is to the limit given by alternative 2 of the theorem.

\begin{figure}
	\centering
	\includegraphics[width=\columnwidth]{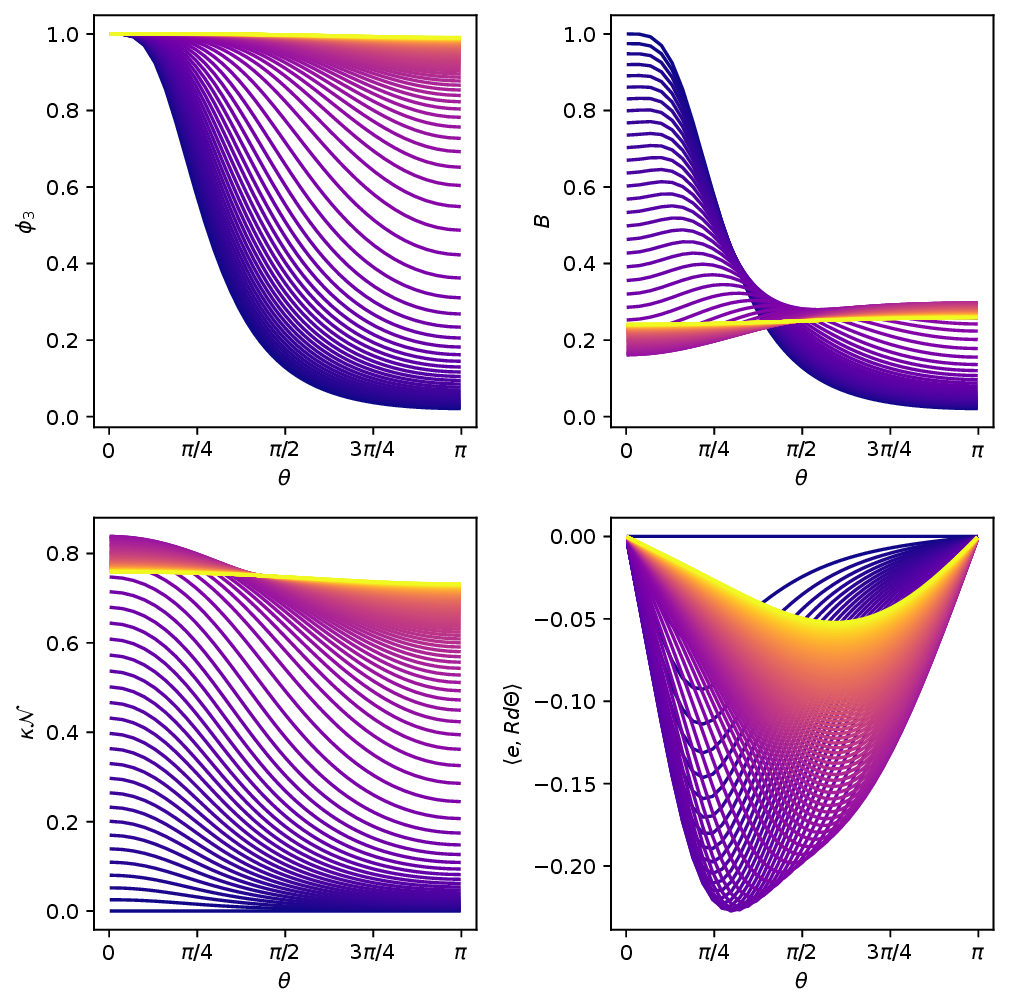}
	\caption{Family of solutions to the Bogomol'nyi equations on a 
		sphere of radius 2 along declination angle $\theta$. 
		Two vortices are located at north pole with no antivortices. \textbf{Top row}. Higgs field north pole projection \(\phi_{3}\) and magnetic field \(B\). \textbf{Bottom row}. Scalar field rescaling by the Chern-Simons constant and electric field in the declination angle direction. 
		The asymmetry parameter was set to $\tau = 0$ and $q = 1$. 
		Pseudo-arclength continuation advances towards lighter colour.}
	\label{fig:sols20}
\end{figure}

\begin{figure}
	\centering
	\includegraphics[width=\columnwidth]{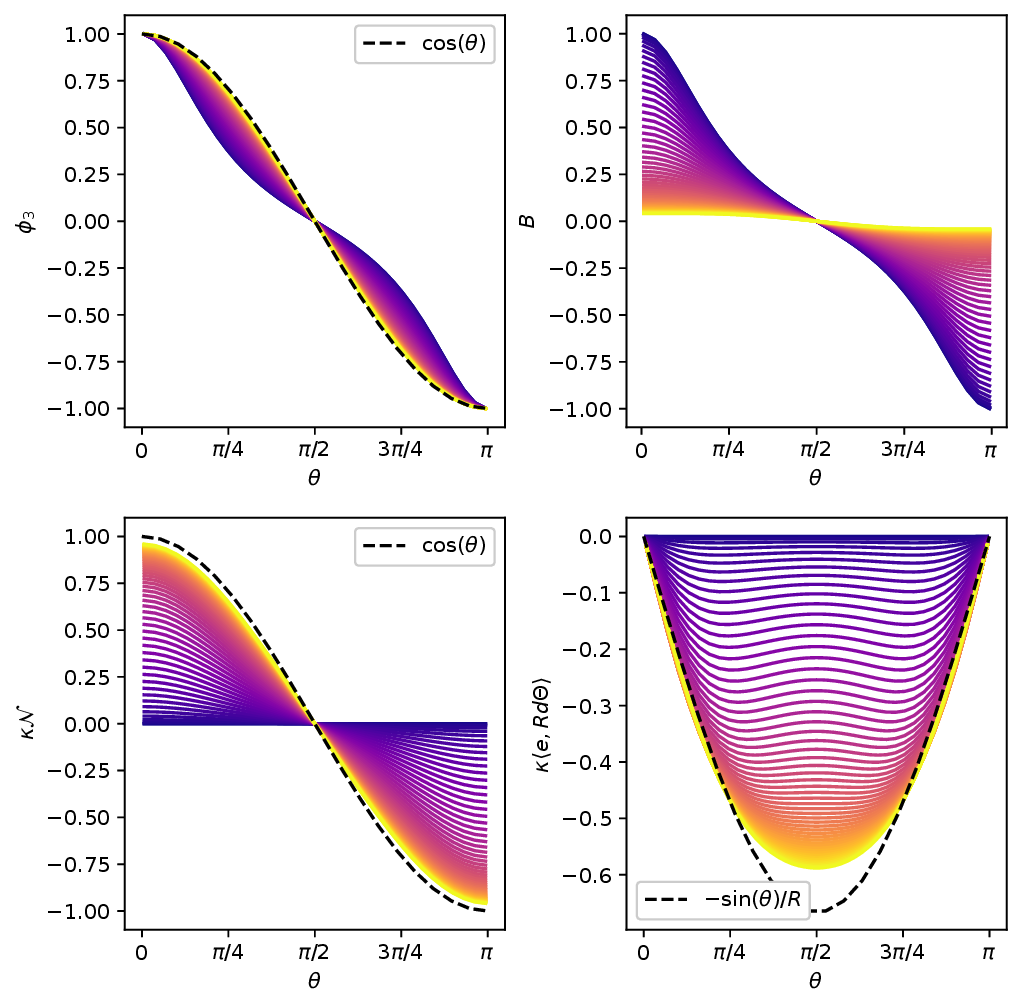}
		
	\caption{Family of solutions to the Bogomol'nyi equations on a 
		sphere of radius $1.5$ along declination angle $\theta$. 
		A vortex is located at north pole and an antivortex at south pole. 
		\textbf{Top row}. Higgs field north pole projection \(\phi_{3}\) and magnetic field \(B\). \textbf{Bottom row}. Scalar field and electric field in the declination angle direction, both rescaled by the Chern-Simons deformation constant \(\kappa\). 
		The asymmetry parameter was set to $\tau = 0$ and $q = 1$. 
		Dashed lines represent the non-trivial limits as $\kappa \to \infty$. 
		Pseudo-arclength continuation advances towards lighter colour.}
	\label{fig:sols11}
\end{figure}

\section*{\large{Data Availability}}

Data sets generated for figures~\ref{fig:sols20} and~\ref{fig:sols11} are available from the corresponding author on reasonable request.

\section*{\large{Acknowledgements}}

René Israel García Lara acknowledges the support by the UNAM Postdoctoral Program (POSDOC) and by SECIHTI-SNI 219537.

\bibliographystyle{plain}
\bibliography{references}

\end{document}